\documentclass[11pt]{article}
\usepackage{odonnell}
\usepackage{etoolbox}
\newcommand{\specrad}{\rho}
\newcommand{\csq}{s_c}
\newcommand{\dsq}{s_d}
\newcommand{\vr}{\varrho}
\newcommand{\vrstar}{\vr^*}
\newcommand{\lupper}{\ol{\lambda}}
\newcommand{\llower}{\ul{\lambda}}
\newcommand{\DEG}{\kappa}
\newcommand{\nontriv}{\mathrm{PS}}
\newcommand{\spec}{\mathrm{spec}} \newcommand{\Spec}{\spec}
\newcommand{\Tree}{\mathbb{T}}
\newcommand{\Primal}{\mathbb{G}}
\newcommand{\vdist}{\mathrm{dist}}
\newcommand{\EIG}{\mathrm{EIG}}
\newcommand{\LTheta}[1]{\vartheta(\overline{#1})}
\newcommand{\SDP}{\mathrm{SDP}}
\newcommand{\SDPtri}{\mathrm{SDP}_\triangle}
\newcommand{\OPT}{\mathrm{OPT}}

\begin{document}

\title{The threshold for SDP-refutation of random regular NAE-3SAT}
\author{Yash Deshpande\thanks{MIT Department of Mathematics. \texttt{yash@mit.edu}} \and Andrea Montanari\thanks{Department of Electrical Engineering and Department of Statistics, Stanford University. \texttt{montanari@stanford.edu}. Supported by grants NSF CCF-1714305 and NSF IIS-1741162.} \and Ryan O'Donnell\thanks{Computer Science Department, Carnegie Mellon University.  \texttt{odonnell@cs.cmu.edu}. Supported by NSF grants CCF-1618679, CCF-1717606. This material is based upon work supported by the National Science Foundation under grant numbers listed above. Any opinions, findings and conclusions or recommendations expressed in this material are those of the author and do not necessarily reflect the views of the National Science Foundation (NSF).} \and Tselil Schramm\thanks{Harvard University and MIT. \texttt{tselil@seas.harvard.edu}. This work was partly supported by an NSF graduate research fellowship (1106400), and took place in part while T.S. was a fellow at the Simons Institute program on Optimization.} \and Subhabrata Sen\thanks{Microsoft Research NE and MIT Department of Mathematics. \texttt{ssen90@mit.edu}}}

\maketitle

\begin{abstract}
    Unlike its cousin 3SAT, the NAE-3SAT (not-all-equal-3SAT) problem has the property that spectral/SDP algorithms can efficiently refute random instances when the constraint density is a large constant (with high probability).  But do these methods work immediately above the ``satisfiability threshold'', or is there still a range of constraint densities for which random NAE-3SAT instances are unsatisfiable but hard to refute?

    We show that the latter situation prevails, at least in the context of random regular instances and SDP-based refutation.  More precisely, whereas a random $d$-regular instance of NAE-3SAT is easily shown to be unsatisfiable (whp) once $d \geq 8$, we establish the following sharp threshold result regarding efficient refutation: If $d < 13.5$ then the basic SDP, even augmented with triangle inequalities, fails to refute satisfiability (whp); if $d > 13.5$ then even the most basic spectral algorithm refutes satisfiability~(whp).
\end{abstract}

\setcounter{page}{0}
\thispagestyle{empty}
\newpage

\section{Introduction}

A randomly chosen $n$-variable constraint satisfaction problem (CSP) will typically be unsatisfiable once the constraint density~$\alpha$ (ratio of constraints to variables) is a sufficiently large constant.  Taking 3SAT as an example, the conjectural satisfiability threshold~\cite{MPZ02,MMZ06} is $\alpha_c \approx 4.2667$, and the trivial first moment method already establishes unsatisfiability (whp) once $\alpha > \log_{7/8}(1/2) \approx 5.19$.  Despite this, there is no known efficient algorithm that can refute random 3SAT instances (whp) for any large constant~$\alpha$.  The best known algorithms~\cite{FGK05,GL03,CGL07,FO07,FKO06}, all of which use spectral or semidefinite-programming (SDP) techniques, work only once $\alpha \gtrapprox \sqrt{n}$.  Indeed, there are lower bounds~\cite{Sch10,Tul09,KMOW17} showing that any polynomial-time algorithm based on such techniques --- more generally, based on the constant-degree ``Sum of Squares'' method --- will fail to refute unless $\alpha \gtrapprox \sqrt{n}$.  The most general of these results~\cite{KMOW17} applies to \emph{any} CSP for which the constraint predicate supports a pairwise-uniform probability distribution.\footnote{That is, there is a distribution $D$ over satisfying assignments $x$ to the predicate, with the property that the order $1$ and $2$ moments of $D$ are identical to those of the uniform distribution.}

On the other hand, for any CSP whose predicate does \emph{not} support a pairwise-uniform probability distribution, it has been shown~\cite{AOW15} that there \emph{is} an efficient SDP-based algorithm for refuting random instances once the constraint density~$\alpha$ is a sufficiently large constant.\footnote{In~\cite{AOW15}, it is stated that $\alpha = n^{k/2-1}\polylog n$ suffices when no $k$-wise uniform distribution is supported; however, in the particular case of $k = 2$ one can show that the $\polylog n$ is unnecessary, using the (worst-case) strong refutation algorithm for 2XOR-SAT~\cite{CW04}.}  For such CSPs, where ``all of the action'' is in the sparse regime of $O(n)$ constraints, it is more plausible to hope for an efficient refutation algorithm that works just above the satisfiability threshold --- or at least to identify sharp thresholds for when efficient refutation algorithms succeed.

Perhaps the simplest and most natural $\mathsf{NP}$-complete CSP of this type is NAE-3SAT.  This is the variant of 3SAT in which a clause is considered ``satisfied'' if and only if it has at least one true literal \emph{and} one false literal; i.e., the literals' truth values are Not All Equal.  (The further variant wherein all literals appear positively is equivalent to the problem of $2$-coloring a $3$-uniform hypergraph.)  Being a more symmetric --- and in some sense, simpler --- variant of 3SAT, the NAE-3SAT problem has received a great deal of attention in the study of random CSPs; see, e.g.,~\cite{AS93,ACIM01,AM02,GJ03,CNRZ03,DRZ08,DKR15,DSS16}.  In particular, by 2003 Goerdt and Jurdzi{\'n}ski~\cite{GJ03} had already proven that SDP methods could refute random NAE-3SAT instances at sufficiently high constant constraint density.  NAE-3SAT is also closely related to the Max-Cut and 2XOR-SAT CSPs and has a natural basic SDP relaxation; for this reason, the problem has also been well-studied from the point of view of worst-case approximation algorithms~\cite{KLP96,AE98,Zwi98,Zwi99}.

This paper is motivated by the question of whether efficient algorithms might be able to refute unsatisfiability of random NAE-3SAT instances at densities all the way down to the satisfiability threshold --- or whether there is still a range of constant densities where random instances are unsatisfiable, but this is hard for efficient algorithms to certify.  The latter case seems to prevail for 3SAT, and one would likely pessimistically guess the same is true for NAE-3SAT.  However one may need a finer analysis for NAE-3SAT; the range of presumably-hard densities for refuting 3SAT is between a constant and~$\sqrt{n}$, whereas for NAE-3SAT it is between two universal constants.

One way to give evidence for the existence of hard densities for NAE-3SAT refutation would be to study the \emph{SDP-satisfiability threshold} for random instances; i.e., the largest density for which the basic SDP algorithm fails to refute satisfiability.  The goal would be to give a lower-bound for the SDP-satisfiability threshold that exceeds the actual NAE-3SAT satisfiability threshold. In fact, the main result of this paper is a determination of the \emph{exact} SDP-satisfiability threshold of random NAE-3SAT instances, in the setting of random regular instances.  This threshold provably exceeds the actual satisfiability threshold, thus establishing a range of degrees for which random regular NAE-3SAT refutation is hard for SDP algorithms.

\subsection{Our results}   \label{sec:our-results}
For technical simplicity, we work in the setting of random \emph{regular} instances of NAE-3SAT, where every variable participates in the same number, $d$, of 3NAE-constraints.  (This is in contrast to the ``Erd\H{o}s--R\'{e}nyi'' setting with clause density~$\alpha$, in which the degree of each variable is like a  Poisson random variable with mean~$3\alpha$.)  We also use the ``random lift'' model for $d$-regular instances, rather than, say, the ``configuration'' model.  For precise details see \Cref{sec:lift-notation}, but in brief, our random $d$-regular instances are chosen as follows:
\begin{enumerate}[label=\roman*]
    \item Start with the bipartite graph $K_{d,3}$.
    \item Choose a uniformly random $n$-lift~$\bH$, a bipartite graph with $dn$ vertices of degree $3$ in one part and $3n$ vertices of degree $d$ in the other part.
    \item Treat the degree-$d$ vertices as CSP variables and the degree-$3$ vertices as 3NAE constraints on the adjacent variables
    \item In each constraint, randomly replace each variable-appearance with its negation, uniformly and independently.
\end{enumerate}
Notice that for \emph{any} $(3,d)$-biregular graph~$H$ and any truth assignment to the variables, the randomness from the negations alone gives us that each constraint is independently satisfied with probability~$3/4$.  Thus the first moment method implies the following:
\begin{fact}                                        \label{fact:sat-thresh-upper}
    For $d > \log_{\frac43} 8 \approx 7.228$ (i.e., for $d \geq 8$) a random $d$-regular NAE-3SAT instance will be unsatisfiable with high probability (indeed, in any model with random negations).\footnote{In fact, the unsatisfiability threshold is more likely to be lower, specifically $d \geq 7$, based on heuristics from statistical physics. The ``1RSB'' prediction for the unsatisfiability threshold of random NAE-3SAT --- which was rigorously verified for NAE-$k$SAT, $k \geq k_0$, in~\cite{DSS16} --- was determined to be at average degree $3 \cdot 2.105 = 6.315$ in the Erd\H{o}s--R\'{e}nyi case~\cite{CNRZ03}, and at degree at most~$7$ in the regular case~\cite{DRZ08} (albeit these predictions were for the ``coloring'' version of NAE-3SAT without negations).}
\end{fact}

Our main theorem is the following sharp threshold for SDP-satisfiability:
\begin{theorem}                                     \label{thm:our-mainest}
    Let $\bI$ be a random $d$-regular instance of NAE-3SAT.  Then with high probability (meaning probability $1-o_{n \to \infty}(1)$):
    \begin{itemize}
        \item For $d < 13.5$,  the natural SDP relaxation will \emph{not} refute satisfiability of~$\bI$.
        \item For $d > 13.5$, the  natural SDP relaxation \emph{will} refute satisfiability of~$\bI$.
    \end{itemize}
\end{theorem}
Of course, since $d$ is always an integer we could have phrased the two cases as $d \leq 13$ and $d \geq 14$.  However, as will be seen below, there is a sense in which the precise non-integer $13.5$ is the sharp threshold.  In any case, these results show that for $d = 8, 9, 10, 11, 12, 13$ (and likely also $d = 7$), a random $d$-regular NAE-3SAT instance is unsatisfiable, yet this cannot be efficiently refuted using the basic SDP relaxation.\\

In fact, our results are somewhat stronger than what is stated in \Cref{thm:our-mainest}.  Let us define
\[
    f(d) = \frac98 - \frac38\cdot\frac{\parens*{\sqrt{d-1} - \sqrt{2}}^2}{d},
\]
a quantity that decreases on $[3, \infty)$, with $f(13.5) = 1$ and $\lim_{d \to \infty} f(d) = 3/4$. We show:
\begin{itemize}
     \item (See \Cref{thm:sdp-for-girth1,thm:sdp-for-girth2} for details.) Even when augmented with the triangle inequalities, the SDP ``thinks'' that a random $d$-regular NAE-3SAT instance has a solution satisfying at least an $f(d) - \eps$ fraction of the constraints; in particular, it thinks the instance is satisfiable if $d < 13.5$.  Indeed this holds for \emph{any} $d$-regular NAE-3SAT instance of sufficiently large constant girth.
     \item (See \Cref{thm:refutation-upper} for details.)  Even the basic ``eigenvalue bound'' (a special case of the SDP method) shows that a random $d$-regular NAE-3SAT instance has no solution satisfying at least an $f(d) + \eps$ fraction of the constraints; in particular, it refutes satisfiability if $d > 13.5$.
\end{itemize}

\section{Methodology, further generalizations, and related work}

\subsection{2XOR-SAT and semidefinite programming}
One reason that semidefinite programming algorithms are particularly natural for NAE-3SAT is that the CSP is essentially a form of 2XOR-SAT.  Recall that the 2XOR-SAT CSP has constraints on pairs of literals, with the constraint being satisfied if the literals are assigned unequal truth values.  Now for literals $\ell_1, \ell_2, \ell_3$:
\begin{align*}
    \text{NAE}(\ell_1,\ell_2,\ell_3) \text{ \phantom{un}satisfied} &\iff \text{exactly $2$ of } \text{XOR}(\ell_1, \ell_2), \text{ XOR}(\ell_2, \ell_3), \text{ XOR}(\ell_3, \ell_1) \text{ satisfied;} \\
    \text{NAE}(\ell_1,\ell_2,\ell_3) \text{ unsatisfied} &\iff \text{exactly $0$ of } \text{XOR}(\ell_1, \ell_2), \text{ XOR}(\ell_2, \ell_3), \text{ XOR}(\ell_3, \ell_1) \text{ satisfied.}
\end{align*}
(In case all the literals are variables appearing positively, the resulting 2XOR-SAT instance is in fact a ``Max-Cut'' instance.) If we convert an NAE-3SAT CSP with $m$ constraints to a 2XOR-SAT CSP with $3m$ constraints in the above way, every truth assignment satisfying a $\beta$ fraction of NAE-3SAT constraints satisfies a $(2/3)\beta$ fraction of 2XOR-SAT constraints.

Indeed, the standard SDP relaxation for NAE-3SAT, first studied by Kann, Lagergren, and Panconesi~\cite{KLP96}, is nothing more than $3/2$ times the basic Goemans--Williamson~\cite{GW95} SDP for the associated 2XOR-SAT instance.  We recall here the basic definitions:
\begin{definition}
    Let $I$ be an instance of 2XOR-SAT with $m$ constraints on $n$ variables, to be assigned values in $\{\pm 1\}$.  We identify the instance with its (multi)set of constraints. Each constraint is a triple $(u,v,\xi)$ for $u,v \in [n]$ distinct and $\xi \in \{\pm 1\}$; this is thought of as the constraint $x_ux_v = -\xi$. The SDP relaxation value is defined to be
    \[
        \SDP(I) = \sup\braces*{ \frac{1}{m} \sum_{(u,v,\xi) \in I} \parens*{\frac12 - \frac12 \xi \la X_u, X_v \ra}} \in [0,1],
    \]
    where the $\sup$ is over all choices of vectors $(X_v)_{v \in [n]}$ satisfying $\la X_v, X_v \ra = 1$ for all~$v$.  Equivalently, instead of vectors, the $X_v$'s may be jointly (centered) Gaussian random variables, with $\la X_u, X_v \ra$ interpreted as $\E[X_u X_v]$.  The quantity $\SDP(I)$ always upper-bounds $\OPT(I)$, the maximum fraction of simultaneously satisfiable 2XOR-SAT constraints, since for any truth assignment $x \in \{\pm 1\}^n$ we may take the joint Gaussians $X_u = x_u Z$, where $Z$ is a standard Gaussian.  The advantage of $\SDP(I)$ is that while computing $\OPT(I)$ is $\mathsf{NP}$-hard, one can compute $\SDP(I)$ (to additive accuracy $2^{-n}$) in polynomial time.
\end{definition}
\begin{definition}
    A common algorithmic technique is to also enforce the \emph{triangle inequalities}, meaning to only take the $\sup$ over $X_v$'s satisfying
    \[
        \la X_u, X_v \ra + \la X_v, X_w \ra + \la X_w, X_u \ra  \geq -1, \qquad \la X_u, X_v \ra - \la X_v, X_w \ra - \la X_w, X_u \ra  \geq -1.
    \]
    The resulting value, $\SDPtri(I)$, is a tighter relaxation: $\OPT(I) \leq \SDPtri(I) \leq \SDP(I)$.
\end{definition}
\begin{definition}
    A related quantity is the \emph{Lov\'{a}sz theta function}~\cite{Lov79}; for a graph~$G$, the Lov\'asz theta function (of its complement), $\LTheta{G}$, is the least~$k$ such that there are centered joint Gaussians $(X_u)$ with $\la X_u, X_u \ra = 1$ for all vertices~$u$ and $\la X_u, X_v \ra  = -\frac{1}{k-1}$ for all edges~$(u,v)$.  In particular, if $G$ is thought of as a Max-Cut instance, then $\SDP(G) \geq \frac12 + \frac12\frac{1}{\LTheta{G}-1}$.
\end{definition}
\begin{definition}
    The SDP for 2XOR-SAT is also known to have a \emph{dual} characterization~\cite{DP93}:
    \[
        \SDP(I) = \inf_{\substack{w \in \R^n \\ \sum_u w_u = 0}} \\  \braces*{\frac{n}{4m} \cdot \lambda_{\text{max}}(L_I + \diag(w))},
    \]
    where $L_I$ denotes the \emph{Laplacian} matrix for~$I$ (defined in \Cref{sec:matrix-notation}), and $\lambda_{\text{max}}$ denotes the largest eigenvalue.  Note that by taking $w = 0$ we get an upper bound on~$\SDP(I)$; we refer to this as the \emph{eigenvalue bound},
    \[
        \EIG(I) =  \frac{n}{4m} \cdot \lambda_{\text{max}}(L_I) =  \frac{1}{2d} \cdot\lambda_{\text{max}}(L_I),
    \]
    the latter equality holding in case $I$ is $d$-regular.  The certificate $\OPT(I) \leq \EIG(I)$ is easy to see; it is a consequence of the definitions that $\OPT(I) = \frac{n}{4m} \cdot \max\{x^\top L_I x : x \in \{\pm \frac{1}{\sqrt{n}}\}^n\}$, and $\lambda_{\text{max}}(L_I)$ allows taking the max over all unit vectors.
\end{definition}

\subsection{Methodology and related work}   \label{sec:prior}
To prove \Cref{thm:our-mainest}, we convert our random NAE3-SAT instances into random 2XOR-SAT instances, and then try to analyze whether or not the SDP-value of these instances is as large as~$\frac23$.  (Recall that this corresponds to the SDP-value of the NAE3-SAT instances being as large as~$1$.)  There are a number of prior works on analyzing the Goemans--Williamson SDP on random graphs (see below);  however, our situation is a bit different.  The main difference is that the graphs underlying our random 2XOR-SAT instances are not uniformly random $2d$-regular graphs, but rather have a peculiar ``triangle-structure''.  Recall that they are generated by first choosing a large random $(3,d)$-biregular graph (by randomly lifting $K_{d,3}$), then replacing each $3$-regular vertex on the left with a triangle on the right.  Thus, locally, the resulting graphs look like the graph on the right in \Cref{fig:infinite-graphs} (for $d = 4$).  An additional small complication is that these random ``triangle-graphs'' effectively get random edge-signings when the random literal-negations are taken into account, converting the Max-Cut instance to a 2XOR-SAT instance.  Finally, in the remainder of the paper we will focus on the generalized problem in which triangles are replaced by $c$-cliques, for $c \geq 3$.  This generalization does not correspond to any well-known CSP, but analyzing general~$c$ turns out to be no harder than analyzing the $c = 3$ special case.

For the part of our main theorem showing that the simple eigenvalue bound succeeds as $d$ becomes large, we need to show tight bounds on the eigenvalues of the random ``triangle-graphs'' (more generally, $c$-clique graphs) that arise in our model.  If we simply had random $d$-regular graphs, Friedman's famous almost-Ramanujan theorem~\cite{Fri03} would have sufficed.  Instead, we relate the eigenvalues of our random graphs to those of a randomly lifted $(c,d)$-biregular bipartite graph.  We then use Bordenave's recent reproof~\cite{Bor17} of Friedman's theorem (revised to also include random edge-signings), as well as the Ihara--Bass formula, to show that with high probability the nontrivial spectrum of such random bipartite graphs is contained in $\pm [\sqrt{d-1} - \sqrt{c-1}, \sqrt{d-1} + \sqrt{c-1}]$.  Inspiration for these computations comes from~\cite{FM16}.

For the part of our main theorem showing that large-value SDP solutions exist, the tools we use come from a fairly recent line of work concerning ``Gaussian waves'' in infinite regular graphs~\cite{Elo09,CGHV15,HV15}.  This work can be seen as giving a way to convert eigenfunctions on the infinite regular tree (and other vertex-transitive infinite graphs) into Goemans--Williamson SDP solutions --- in fact, Lov\'{a}sz theta function solutions.  These may be converted to such solutions on high-girth finite graphs that locally resemble the infinite graphs.  Several works in this area~\cite{CGHV15,HV15,Cso16,Lyo17} used this method to show, e.g., that high-girth $3$-regular graphs must contain large independent sets, using techniques resembling the randomized rounding of independent-set SDPs (cf.~\cite{KMS98}) and also local improvement techniques applicable to cubic graphs (cf.~\cite{HLZ04}).  These techniques were also used to show limits on the performance of SDP for Max-Cut, Min-Bisection, and community detection problems in, e.g.,~\cite{MS16,FM16}.  See~\cite{BKM17} for similar approaches in the context of graph-coloring, and~\cite{JMR16} for more on phase transitions for SDPs in the context of community detection.

\section{Preliminaries on graphs, lifts, and eigenvalues}

\subsection{Graphs, hypergraphs, and edge-labeled graphs}       \label{sec:graph-notation}
We begin with some general notation.

$H$ will typically denote a simple $(c,d)$-biregular bipartite graph with $c, d \geq 2$.  The setting of most interest to us is $d \geq c = 3$.  Sometimes we will refer to the vertices on the $c$-regular side as \emph{constraints} and the vertices on the $d$-regular side as \emph{variables}. \Cref{fig:k43} shows an example, $K_{4,3}$, with the variables depicted as circles and the constraints depicted as squares.

\myfig{.1}{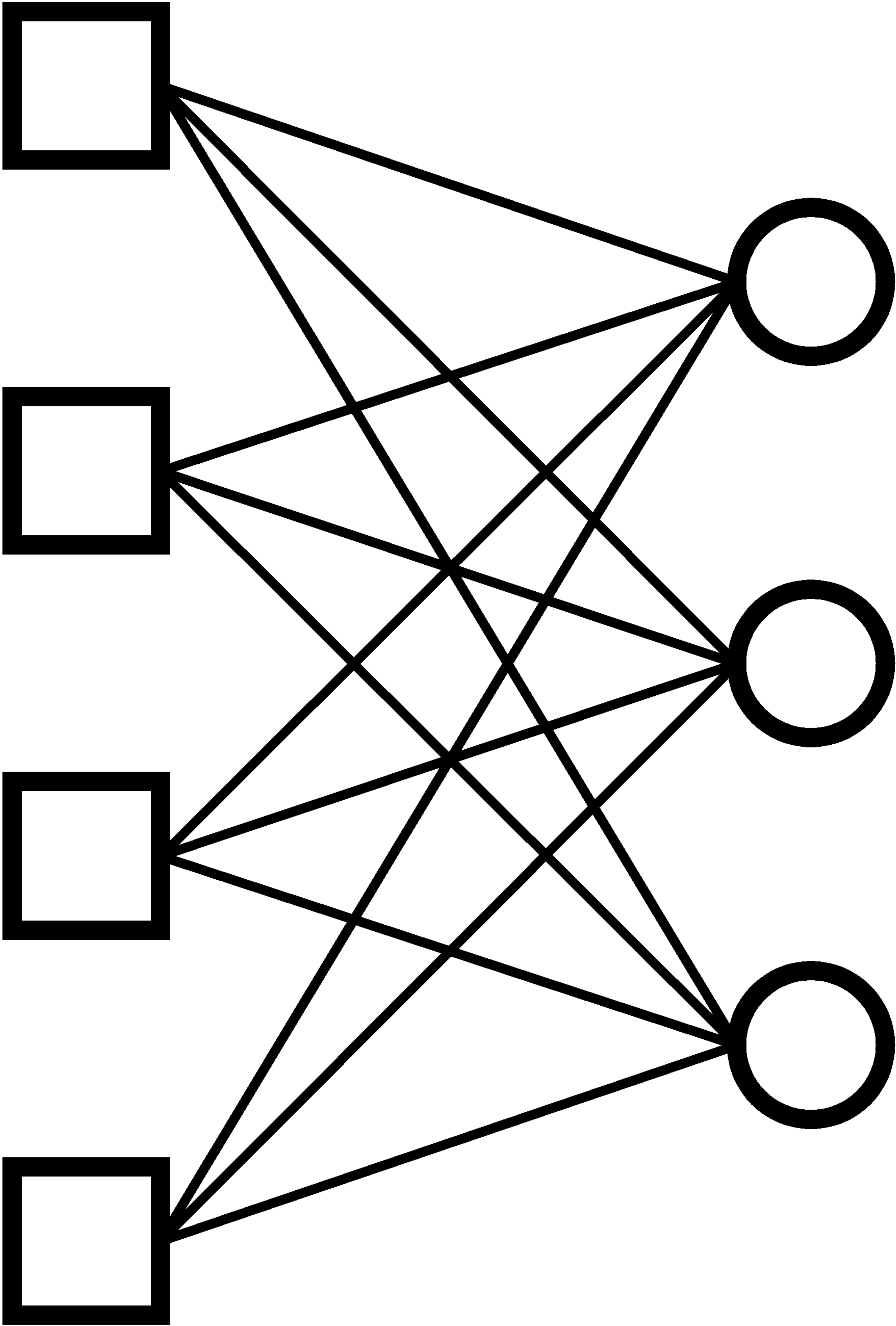}{$H = K_{4,3}$}{fig:k43}

We may also think of $H$ as a $c$-uniform $d$-regular hypergraph, with the variables as vertices and constraints as hyperedges.  $X$ will denote an edge-signed version of $H$ (thought of as a bipartite graph, not a hypergraph); i.e., one in which each edge of~$H$ is labeled with~$\pm 1$.  (In the unsigned case, we think of all edges as being labeled~$+1$.)  We say that $X$ is a ``random signing'' of $H$ if it is formed by independently labeling each edge of $H$ with~$\pm 1$, uniformly at random.

Given $H$, we will write $G = G_H$ for the (loopless multi-)graph formed by first thinking of $H$ as a hypergraph and then replacing each hyperedge by a $c$-clique.  As a result, $G$ is a $(c-1)d$-regular graph, called the \emph{primal graph} for~$H$.   Given an edge-signed version $X$ of $H$, we will write $I = I_X$ for the primal graph of~$X$, an edge-signed version of $G$ defined as follows:  whenever constraint $a$ is adjacent to variables~$i, j$ with edge-signs $\xi_{ai}, \xi_{aj} \in \{\pm 1\}$, we place the sign $\xi_{ai}\xi_{aj}$ on the resulting $\{i,j\}$ edge of~$G$.  We may think of~$I$ as a 2XOR-SAT instance, where the vertices are to be assigned values $x_i \in \{\pm 1\}$, and an edge $\{i,j\}$ with label $\xi$ corresponds to the constraint $x_ix_j = -\xi$.

In the special case of $c=3$, we can think of~$X$ as a NAE-3SAT instance, where the variables are to be assigned values $x_i \in \{\pm 1\}$, and a constraint $a$ adjacent to variables $i,j,k$ with labels $\xi_{ai}, \xi_{aj}, \xi_{ak}$ corresponds to the constraint that $\xi_{ai} x_i, \xi_{aj} x_j, \xi_{ak} x_k$ are not all equal.  In this case there is a precise relationship between the NAE-3SAT instance~$X$ and the 2XOR-SAT instance~$I$; any assignment to the vertices satisfying exactly a $\beta$ fraction of the NAE-3SAT constraints will necessarily satisfy exactly a $\frac23\beta$ fraction of the 2XOR-SAT constraints.

\subsection{Associated matrices}    \label{sec:matrix-notation}
Given any of $Y \in \{H, X, G, I\}$, we will write $A_Y$ for the adjacency matrix.  More precisely, $A_Y[i,j]$ is the sum of the (positive and negative) edge-labels on all edges connecting~$i$ and~$j$.

We will write $D_Y$ for the diagonal degree matrix of~$Y$, whose entry $D_Y[i,i]$ equals the degree of vertex~$i$. (Both signed and unsigned edges count~$1$ toward the degree.)
We write $L_Y = D_Y - A_Y$ for the Laplacian matrix of~$Y$; we also write $L_Y(u) = (1-u^2) \Id + u^2 D_Y - uA_Y$ for the ``deformed Laplacian'', parameterized by $u \in \R$, which reduces to the basic Laplacian when $u = 1$.  (Here $\Id$ denotes the identity operator.)

Finally, we will write $B_Y$ for the non-backtracking matrix of~$Y$.  Recall that this matrix is formed as follows: First, each undirected edge in $Y$ is converted to two directed edges (both having the same sign, in case $Y$ is edge-signed).   Then $B_Y$ is the square (non-symmetric) matrix indexed by the directed edges, in which $B_Y[(i,j),(k,\ell)]$ entry is nonzero if and only if $j = k$ and $i \neq \ell$, in which case it equals the sign-label of~$(i,j)$.

\subsection{Lifts}      \label{sec:lift-notation}
Suppose now that $Y = (V,E)$ denotes any undirected (multi-)graph.  For $n \in \Z^+$, an $n$-lift of $Y$ is a graph $Y_n$ whose vertex set is $V \times [n]$ and whose edges consist of a perfect matching between $\{u\} \times [n]$ and $\{v\} \times [n]$ for each edge $\{u,v\} \in E$.  When the $|E|$ perfect matchings are chosen independently and uniformly at random, we call $Y_n$ a random $n$-lift of $Y$.  Note that if $Y$ is a $d$-regular graph, then so is $Y_n$, and if $Y$ is a $(c,d)$-biregular bipartite graph, then so is $Y_n$.  If $B$ (respectively, $B_n$) denotes the non-backtracking matrix of~$Y$ (respectively, $Y_n)$, it is known that the multiset of $B_n$'s eigenvalues contains the multiset of $Y$'s eigenvalues.  The remaining eigenvalues are referred to as the ``new'' eigenvalues of~$B_n$.

\subsection{Eigenvalues}    \label{sec:eig-notation}
Given an $N$-dimensional matrix $M$, we write $\spec(M) \subset \C$ for its spectrum, the cardinality-$N$ \emph{multiset} of roots of its characteristic polynomial.  We also write $\specrad(M)$ for its spectral radius, $\max\{|\lambda| : \lambda \in \spec(M)\}$.  The adjacency matrix of a (possibly edge-signed) graph is symmetric, and hence its spectrum is real; the Laplacian is furthermore positive semidefinite, and hence its spectrum is nonnegative.  A non-backtracking matrix, however, will in general have complex spectrum.

We are particularly interested in bipartite graphs, so we record some facts concerning them here.  Suppose $X$ is a possibly edge-signed bipartite graph, with vertex parts of size $m \geq n$.  Then it is well known that
\[
    \spec(A_X) = \{0 : \text{with multiplicity } m-n\} \cup \{\pm \lambda : \lambda \in \nontriv(A_X)\}
\]
for some multiset $\nontriv(A_X) \subset \R^{\geq 0}$.\footnote{We chose ``$\nontriv$'' to stand for Positive Spectrum, notwithstanding our warning that it may contain~$0$.}   Further, if $X$ is $(c,d)$-biregular, we'll have $\nontriv(A) \subset [0, \sqrt{cd}]$.  The set $\pm \nontriv(A_X)$ may be called the ``nontrivial'' part of $A_X$'s spectrum.  A warning, though:  $\pm \nontriv(A_X)$ is not the same as the ``nonzero'' part of $A_X$'s spectrum, since $\nontriv(A_X)$ may contain~$0$ with positive multiplicity.  Indeed, this happens in one of the simplest cases, as is well known:
\begin{fact}                                        \label{fact:Kcd-A-spec}
    Let $H = K_{d,c}$, the complete bipartite graph with vertex parts of size $d \geq c$.  Then $\nontriv(A_H)$ consists of $c-1$ copies of~$0$ and $1$ copy of $\sqrt{cd}$.
\end{fact}
We also record below the spectrum of the non-backtracking matrix of $K_{d,c}$, which we'll derive in \Cref{sec:ihara--bass} using the Ihara--Bass formula.  But first, some notation we'll use heavily in this paper:
\begin{notation}    \label{not:cd}
    For $c, d \geq 2$, we write
    \[
        \csq = \sqrt{c-1}, \quad \dsq = \sqrt{d-1}, \quad \rho_1 = \csq \dsq, \quad \lupper = \dsq + \csq, \quad \llower = |\dsq - \csq|, \quad \DEG = (c-1)d = \rho_1^2 + \csq^2.
    \]
    We will often assume $d \geq c$, in which case $\llower = \dsq - \csq$.
\end{notation}
\begin{proposition}                                        \label{prop:Kcd-B-spec}
    Let $B$ be the non-backtracking matrix of $K_{d,c}$, where $d \geq c \geq 2$, $d \neq 2$. Let $i$ be the fourth primitive root of unity. Then
    \[
        \spec(B) = \begin{cases}
                                                       \pm 1& \text{with multiplicity $(c-1)(d-1)$ each;}\\
                                                       \pm i\csq & \text{with multiplicity $(d-1)$ each;}\\
                                                       \pm i\dsq & \text{with multiplicity $(c-1)$ each;}\\
                                                       \pm \csq\dsq & \text{with multiplicity $1$ each;}\\
                                                   \end{cases}\\
        \qquad  \text{and hence, } \specrad(B) = \csq \dsq = \rho_1.
    \]
\end{proposition}

As described in \Cref{sec:graph-notation}, we will often consider forming the primal graph~$G$ of a $(c,d)$-biregular graph~$H$.  It is simple to work out the relationship between the eigenvalues of~$H$ and the eigenvalues of~$G$; this is done in, e.g.,~\cite[Section~4.1]{LS96}.  The analysis is unchanged for the edge-signed variant, and it yields:
\begin{proposition}                                     \label{prop:triangle-replace-eigs}
    Let $X$ be an edge-signed $(c,d)$-biregular graph, and let $I = I_X$ be the corresponding edge-signed primal graph.  Then
    \[
        \Spec(A_I) = \{\lambda^2 - d : \lambda \in \nontriv(A_X)\}.
    \]
    Since $I$ is $\DEG$-regular, where $\DEG = cd - d$, we can also conclude that
    \[
        \Spec(L_I) = \{cd - \lambda^2 : \lambda \in \nontriv(A_X)\}.
    \]
\end{proposition}

\subsection{The infinite biregular tree and distance-regular graph}    \label{sec:infinite-bireg-notation}
Since a large random $(c,d)$-biregular graph looks locally like a tree, we will want to study the infinite $(c,d)$-biregular tree, which we denoted by $\Tree_{d,c}$.  More to the point, we will want to study its (infinite) primal graph, which we denote by $\Primal_{d,c}$. Fragments of these graphs, in the case $c = 3$, $d = 4$, are pictured in \Cref{fig:infinite-graphs}.
\begin{figure}[H]
    \centering
    \includegraphics[width=.3\textwidth]{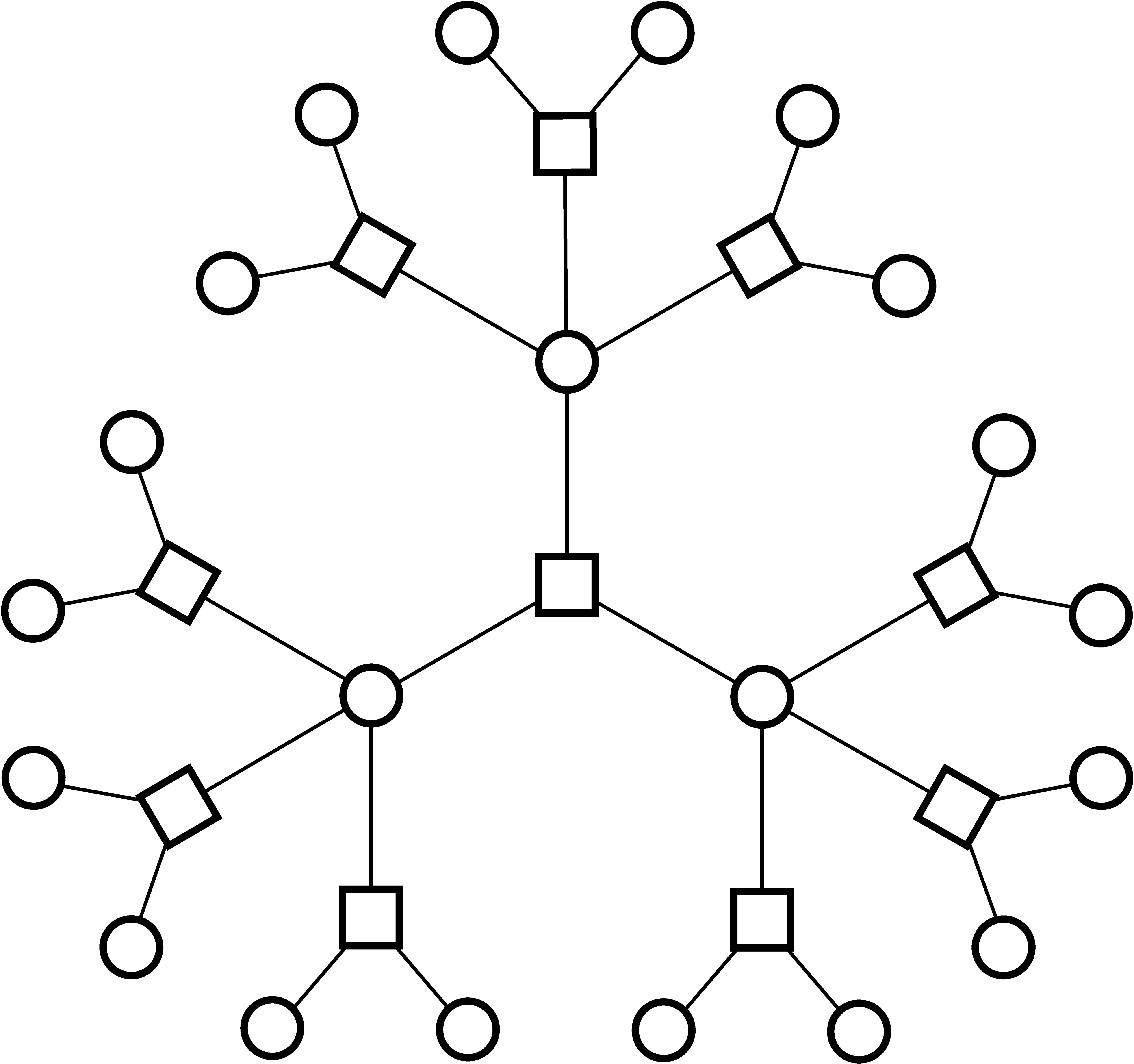}
    \qquad \qquad
    \includegraphics[width=.3\textwidth]{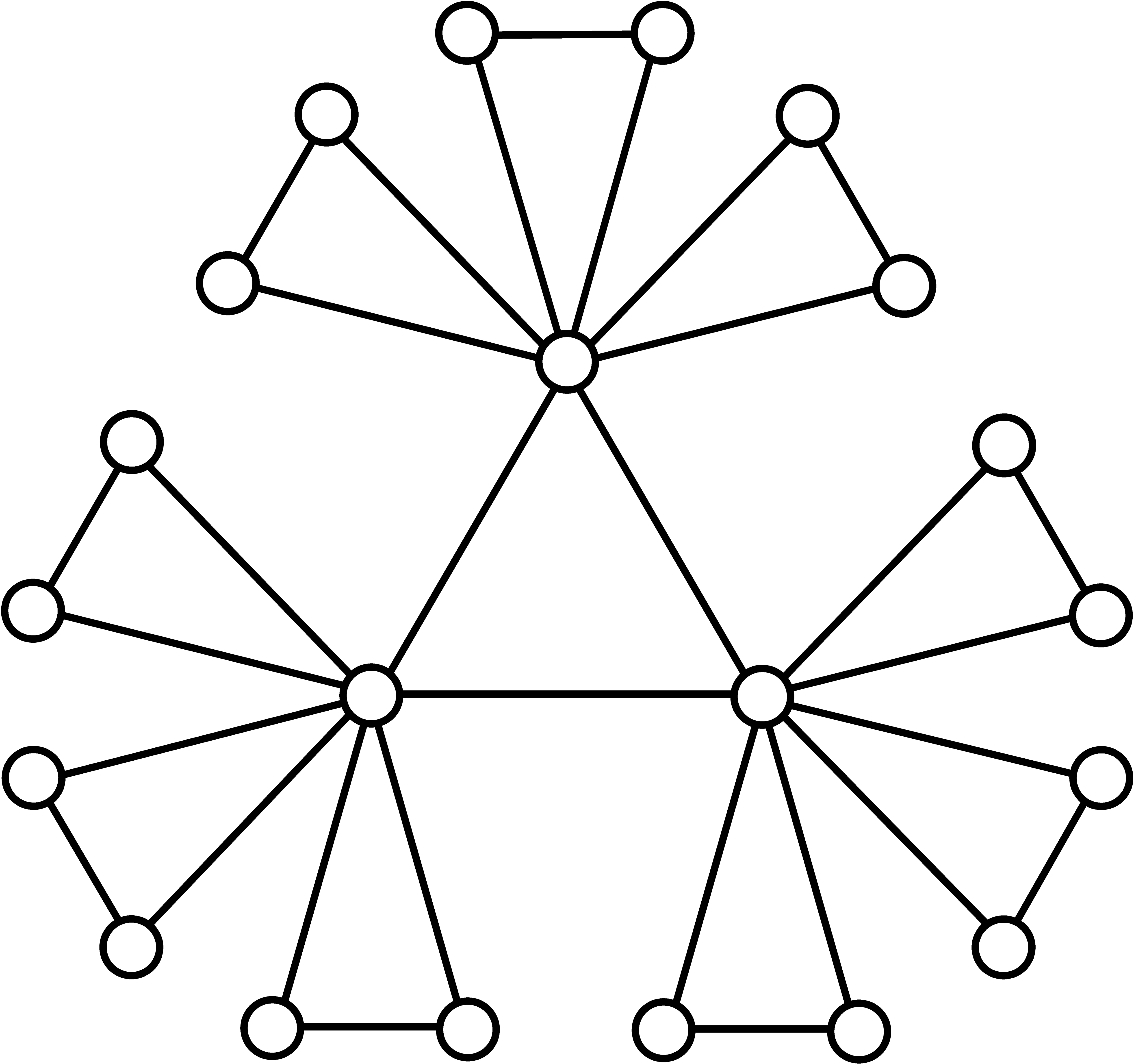}
    \caption{Fragments of the infinite biregular tree $\Tree_{4,3}$, and its primal graph $\Primal_{4,3}$}
    \label{fig:infinite-graphs}
\end{figure}
As shown by Ivanov~\cite{Iva83}, the graphs $\Primal_{d,c}$ are precisely the infinite graphs~$G$ that are \emph{distance-regular}, meaning that there exist constants~$p^h_{j,k}$ such that for every pair $u,v \in V(G)$ with $\dist_G(u,v) = h$, the number of vertices $w \in V(G)$ having $\dist_G(w,u) = j$ and $\dist_G(w,v) = k$ is equal to~$p^h_{j,k}$.  It is elementary to compute these quantities for~$\Primal_{d,c}$, and the results appears below. Only the cases $h = 0, 1$ are truly essential for the paper, and the reader might like to verify them while referring to \Cref{fig:infinite-graphs}.
\begin{proposition}                                     \label{prop:pijk}
    In the distance-regular graph $\Primal_{d,c}$, recalling the notation
    \[
        \csq^2 - 1 = c-2, \quad \rho_1^2 = (c-1)(d-1), \quad \rho_1^2 + \csq^2 = \DEG = (c-1)d, \quad \rho_1^2 - \csq^2 = (c-1)(d-2),
    \]
    we have
    \[
        p^0_{\ell,\ell} = \begin{cases}
                                        1                                                             & \text{if }\ell = 0\\
                                        (\rho_1^2 + \csq^2)\rho_1^{2(\ell-1)}    & \text{if }\ell \geq 1;
                                    \end{cases}
    \]
    and, for $h \geq 1$, $0 \leq t \leq h$,
    \begin{align*}
        \text{if $h$ and $t$ have the same parity:} \quad  p^{h}_{\ell, \ell + t} = p^{h}_{\ell + t, \ell} &=
                                    \begin{cases}
                                        0                                                             & \text{if }\ell < \frac{h-t}{2}\\
                                        1                                                             & \text{if }\ell = \frac{h-t}{2}\\
                                        \rho_1^{2\ell} & \text{if }\ell > \frac{h-t}{2} \text{ and } t = h \\
                                        (\rho_1^2 - \csq^2)\rho_1^{2(\ell-(\frac{h-t+2}{2}))}    & \text{if }\ell > \frac{h-t}{2} \text{ and } t \neq h;
                                    \end{cases} \\
        \text{if $h$ and $t$ have opposite parity:} \quad   p^{h}_{\ell, \ell + t} = p^{h}_{\ell + t, \ell} &=
                                    \begin{cases}
                                        0                                                             & \text{if }\ell < \frac{h-t+1}{2}\\
                                        (\csq^2-1)\rho_1^{2(\ell-\frac{h-t+1}{2})}    & \text{if }\ell \geq \frac{h-t+1}{2};
                                    \end{cases}
    \end{align*}
    and finally, $p^h_{j,k} = 0$ otherwise.
\end{proposition}

The spectrum of the adjacency ``matrix'' (operator) of $\Primal_{d,c}$ --- and indeed, the whole ``spectral measure'' --- has been known since the early '80s.  (There are appropriate definitions for these terms, generalizing the definitions in the finitary case.  We will not give them here since, strictly speaking, this paper does not rely on them.)  In particular,
\begin{equation}    \label{eqn:infinite-spectra}
    \spec(A_{\Tree_{d,c}}) = \{0\} \cup \pm [\llower, \lupper], \qquad \spec(A_{\Primal_{d,c}}) = [\llower^2 - d, \lupper^2 - d];
\end{equation}
(the latter holding under the assumption $d \geq c$; if $d < c$ then also $-d \in \spec(A_{\Primal_{d,c}})$).  The history of these results can be found in \cite[Section~7E]{MW89} and~\cite[Section~5.2]{GM88}, the latter of which also shows that the spectral measures of large random $(c,d)$-biregular graphs converge to a measure with support $\spec(A_{\Tree_{d,c}})$ (and similarly for their primal graphs and $\spec(A_{\Primal_{d,c}})$).

\section{Eigenvalues of random lifts and signings}

Generalizing Friedman's celebrated characterization of the spectrum of random $d$-regular random graphs \cite{Fri03}, Bordenave recently proved the following theorem:
\begin{theorem} (\cite[Theorem 20]{Bor17}.)                                    \label{thm:bordenave}
    Let $Y$ be a connected multigraph (with more edges than vertices) having non-backtracking matrix $B$.  Fix $\eps > 0$.  Let $\bY_n$ be a random $n$-lift of $Y$, and let $\bB_n$ be its non-backtracking matrix.  Then
    \[
        \Pr[\text{$\bB_n$ has a \emph{new} eigenvalue of magnitude} \geq \sqrt{\specrad(B)} + \eps] = o_{n \to \infty}(1).
    \]
\end{theorem}

We will need a variant of this theorem in which the graph is randomly lifted and then randomly signed. The statement and proof are actually a little bit simpler.
\begin{theorem}                                   \label{thm:signed-bordenave}
    Let $Y$ be a connected graph (with more edges than vertices) having non-backtracking matrix $B$.  Fix $\eps > 0$.  Let $\bX_n$ be a random signing of a random $n$-lift $\bY_n$ of $Y$, and let $\bB_n$ denote the  non-backtracking matrix of $\bX_n$.   Then
    \[
        \Pr[\specrad(\bB_n) \geq \sqrt{\specrad(B)} + \eps] = o_{n \to \infty}(1).
    \]
\end{theorem}
The proof, which closely follows that of~\cite[Theorem 20]{Bor17}, appears in \Cref{sec:bordenave}.\\

We will also quote some basic results about the scarcity of cycles in randomly lifted graphs:
\begin{theorem}                                     \label{thm:girthy}
    (Greenhill--Janson--Ruci{\'n}ski~\cite[Lemma~5.1]{GJR10}.)  Let $\bY_n$ be as in \Cref{thm:bordenave} or \Cref{thm:signed-bordenave}, and write $\bZ_k$ for the number of length-$k$ cycles in $\bY_n$.  Let $\bP_2, \bP_3, \dots$ be independent Poisson random variables with $\bP_k$ of mean $w_k/(2k)$, where $w_k = \tr(B^k)$ is the number of closed non-backtracking walks in~$Y$.  Then for any $g \in \N^+$, the random variables $(\bZ_2, \bZ_3, \dots, \bZ_g)$ converge jointly in distribution to $(\bP_2, \bP_3, \dots, \bP_g)$.  In particular, for a fixed~$g$ and $n$ sufficiently large, there is a positive probability (depending only on~$g$ and~$Y$) that $\bY_n$ has girth exceeding~$g$.
\end{theorem}
\begin{theorem}   (Easily extracted from the proof of \cite[Lemma~24]{Bor17}.)   \label{thm:cycleless-neighborhoods}
    Let $\bY_n$ be as in \Cref{thm:bordenave} or \Cref{thm:signed-bordenave} and write $d$ for the maximum degree of~$Y$.  Call a vertex of $\bY_n$ \emph{$g$-bad} if its distance-$g$ neighborhood contains a cycle.  Then the expected number of $g$-bad vertices in~$\bY_n$ is $O((d+1)^g)$.
\end{theorem}

\subsection{The Ihara--Bass formula}            \label{sec:ihara--bass}
The Ihara--Bass formula relates the eigenvalues of a graph's adjacency matrix and its non-backtracking matrix.  Originally proved by Ihara~\cite{Iha66} for regular graphs, it was subsequently generalized to irregular graphs~\cite{Has92,Bas92,ST96,KS00}, vertex-weighted graphs~\cite{Kem16}, and most generally, edge-weighted graphs~\cite{WF09,FM16}.  We will need the last of these, but only in the special case that all edge-weights are~$\pm 1$.  In this case, the resulting formula looks identical to the usual (irregular, unweighted) Ihara--Bass formula:
\begin{theorem} (\cite[Theorem~2]{WF09}, specialized to all edge-weights $\pm 1$.)
\label{thm:ihara-bass1}
    Let $X$ be a edge-signed graph, having adjacency matrix~$A$, non-backtracking matrix~$B$, and deformed Laplacian $L(u) = (1-u^2) \Id + u^2 D - uA$. Then for all real $u \neq \pm 1$,
    \[
        \det(\Id - uB) = \det(L(u)) \cdot (1-u^2)^{\#E(X) - \#V(X)}.
    \]
\end{theorem}
In the special case when $X$ is $(c,d)$-biregular, one can use this formula to work out a very explicit mapping between the eigenvalues of~$A$ and the eigenvalues of~$B$.  The computations appear in~\cite[Section~4.2]{Kem16}; that paper only considered unsigned edges, but the result is the same because the Ihara--Bass formula is identical.  Recalling the notation from \Cref{sec:eig-notation}:
\begin{theorem} (Follows from \cite[Theorem~6]{Kem16} using \Cref{thm:ihara-bass1}.)            \label{thm:ihara-bass2}
    Let $X$ be an edge-signed $(c,d)$-biregular graph, with $m$ vertices on the $c$-regular side and $n$ vertices on the $d$-regular side, so $e = cm = dn$ is the number of edges. Let $A$ denote the adjacency matrix of~$X$.  Then~$B$, the non-backtracking matrix of~$X$, has the following $2e$ eigenvalues:
    \begin{itemize}
        \item $e - (m+n)$ copies each of $\pm 1$.
        \item $m - n$ copies each of $\pm i\csq$.
        \item $4n$ ``nontrivial'' eigenvalues, all roots of $p_\lambda(u) = u^4 + (\csq^2+\dsq^2 - \lambda^2)u^2 + \rho_1^2$ for $\lambda \in \nontriv(A)$.
    \end{itemize}
\end{theorem}

We would now like to understand the location of the $4$ roots of $p_\lambda(u)$ in~$\C$ as $\lambda$ varies in~$[0,\sqrt{cd}]$.
\ignore{
    Since $p_\lambda$ is just a quadratic in~$u^2$, one can easily work out the following picture.  When $\lambda = \sqrt{cd}$, the four roots are at $\pm \rho_1$ and $\pm 1$.  As $\lambda$ decreases, the roots travel along the real axis towards $\pm \sqrt{\rho_1}$, reaching there simultaneously when $\lambda = \lupper$.  As $\lambda$ further decreases, the $4$ roots travel separately in the $4$ quadrants, tracing out the circle of radius $\sqrt{\rho_1}$ until they reach the complex axis at $\pm \sqrt{\rho_1} i$ when $\lambda = \llower$. As $\lambda$ decreases from $\llower$, the pairs split up and travel along the complex axis, ending at final positions $\pm i\csq$ and $\pm i\dsq$ when $\lambda = 0$.
}
To do this, write
\[
    \csq = \frac{\lupper - \llower}{2}, \quad \dsq = \frac{\lupper + \llower}{2}, \quad \alpha = \frac{\lambda^2 - \llower^2}{2}, \quad \beta = \frac{\lambda^2  - {\lupper}^2}{2}, \quad U = u^2.
\]
Then
\[
    p_\lambda(u) = U^2 - (\alpha+\beta) U + \parens*{\frac{\alpha - \beta}{2}}^2,
\]
which has roots
\[
    U = \frac12 \parens*{\sqrt{\alpha} \pm \sqrt{\beta}}^2.
\]
If $\llower^2 \leq \lambda^2  \leq \lupper^2$ then $\beta \leq 0 \leq \alpha$ and
\[
    |U| = \frac12\parens*{\sqrt{\alpha}^2 + \sqrt{-\beta}^2} =\frac{\alpha - \beta}{2} = \frac{\lupper^2 - \llower^2}{4} = \csq\dsq = \rho_1.
\]
On the other hand, if $\lambda^2 \not \in \bracks*{\llower^2, \lupper^2}$, then $\alpha$ and $\beta$ have the same sign and
\[
    |U| = \frac12(|\alpha| + |\beta| \pm 2|\alpha|\cdot |\beta|), \text{ the larger of which exceeds } \frac{\lupper^2 - \llower^2}{4} = \rho_1.
\]
We conclude:
\begin{proposition}                                     \label{prop:the-picture}
    For real $\lambda$, the roots of $p_\lambda(u)$ simultaneously have magnitude at most $\sqrt{\rho_1}$ if and only if $\lambda^2 \in \bracks*{\llower^2, \lupper^2 }$ (i.e., $\lambda \in \pm \bracks*{\llower, \lupper}$).
\end{proposition}

Also, when $\lambda = 0$ we have $p_\lambda(u) = u^4 + (\csq^2 + \dsq^2)u^2 + \csq^2\dsq^2$, and when $\lambda = \sqrt{cd}$ we have $p_\lambda(u) = u^4 - (\rho_1^2 + 1)u^2 + \rho_1^2$.  Thus we can directly verify:
\begin{proposition}                                     \label{prop:lambda-cd}
    For $\lambda = 0$, the $4$ roots of $p_\lambda(u)$ are $\pm i\csq$, $\pm i\dsq$.
    And, for $\lambda = \sqrt{cd}$, the $4$ roots of $p_\lambda(u)$ are $\pm \rho_1$, $\pm 1$.
\end{proposition}
At this point, we can combine \Cref{thm:ihara-bass2}, \Cref{fact:Kcd-A-spec}, and \Cref{prop:lambda-cd} to obtain \Cref{prop:Kcd-B-spec} as stated in \Cref{sec:eig-notation}.  We may furthermore put together all the results in this section:
\begin{theorem}                                     \label{thm:random-signed-lift-eigenvalues}
    Let $d \geq c \geq 2$, $d \neq 2$.  Fix $\eps > 0$.  Let $\bX_n$ be a random signing of a random $n$-lift of the complete bipartite graph $K_{d,c}$, and let $\bA_n$ denote its adjacency matrix.  Then
    \[
        \Pr\bigl[\nontriv(\bA_n) \not \subset  [\llower -\eps, \lupper + \eps]\bigr] = o_{n \to \infty}(1).
    \]
\end{theorem}
\begin{proof}
    We apply \Cref{thm:signed-bordenave} with $Y = K_{d,c}$ and some sufficiently small $\eps' = \eps'(\eps, c, d) > 0$.  The non-backtracking matrix $B$ of~$Y$ has spectral radius $\rho_1$, by \Cref{prop:Kcd-B-spec}.  Thus if $\bB_n$ is the non-backtracking matrix of the randomly signed random lift~$\bX_n$ of~$Y$, we get
    \[
        \Pr\bigl[\specrad(\bB_n) \geq \sqrt{\rho_1} + \eps'\bigr] = o_{n \to \infty}(1),
    \]
    Thus with probability $1 - o(1)$ we have $\specrad(\bB_n) < \sqrt{\rho_1} + \eps'$.  In this case, taking $\eps'$ sufficiently small and using the fact that the roots of a polynomial are continuous in its coefficients, \Cref{prop:the-picture} and \Cref{thm:ihara-bass2} imply that $\nontriv(\bA_n) \subset [\llower -\eps, \lupper + \eps]$.  The proof is complete.
\end{proof}
\begin{remark}
    This theorem is ``to be expected'' in light of the Godsil--Mohar work on spectral convergence mentioned at the end of \Cref{sec:infinite-bireg-notation}.  But of course one needs the hard work of Bordenave's Theorem to show that random $(c,d)$-biregular graphs typically do not \emph{any} eigenvalues outside the spectral bulk.  In fact, to emphasize that care is needed, we remark that the random signing in \Cref{thm:random-signed-lift-eigenvalues} is essential; without it, it's not hard to show that $\nontriv(\bA_n)$ will contain~$0$ with probability~$1$.
\end{remark}
\begin{corollary}                                       \label{cor:random-triangle-graph-eigenvalues}
    Let $d \geq c \geq 2$, $d \neq 2$.  Fix $\eps > 0$.  Let $\bX_n$ be a random signing of a random $n$-lift of the complete bipartite graph $K_{d,c}$, let $\bI_n$ be the associated 2XOR-SAT instance (as in \Cref{sec:graph-notation}), and let $\bL_n$  be its Laplacian matrix.  Then
    \[
        \Pr\bigl[\bL_n \text{ has an eigenvalue outside } [(1-\rho_1)^2 - \eps, (1+\rho_1)^2  +\eps]\bigr] = o_{n \to \infty}(1).
    \]
\end{corollary}
\begin{proof}
    This follows from \Cref{prop:triangle-replace-eigs}, $cd - \lupper^2 = (1-\rho_1)^2$,  and $cd - \llower^2 = (1+\rho_1)^2$.
\end{proof}

\Cref{cor:random-triangle-graph-eigenvalues} now directly implies the following:
\begin{theorem}                                     \label{thm:refutation-upper}
    Let $d \geq c \geq 2$, $d \neq 2$.  Fix $\eps > 0$. Let $\bI_n$ be a random 2XOR-SAT instance as in \Cref{cor:random-triangle-graph-eigenvalues}, so $\bI_n$ is $\DEG$-regular ($\DEG = (c-1)d$) with $cn$ variables and $\binom{c}{2}dn$ constraints.  Then
    \[
        \Pr\bracks*{\EIG(\bI_n) \geq \frac{(1+\rho_1)^2}{2\DEG} + \eps} = o_{n\to \infty}(1),
    \]
    where $\rho_1 = \sqrt{c-1}\sqrt{d-1}$.

    In case $c = 3$, if we view $\bI_n$ as a random $d$-regular NAE-3SAT instance on~$3n$ variables (chosen according to the random lift/sign model), we have
    \[
        \Pr\bracks*{\EIG(\bI_n) \geq \frac98 - \frac38\cdot\frac{\parens*{\sqrt{d-1} - \sqrt{2}}^2}{d} + \eps} = o_{n\to \infty}(1),
    \]
\end{theorem}
As mentioned in \Cref{sec:our-results}, the quantity $\frac98 - \frac38\cdot\frac{\parens*{\sqrt{d-1} - \sqrt{2}}^2}{d}$ decreases from $\frac98$ to $\frac34$ on~$[3,\infty)$ and takes value~$1$ at $d = 13.5$. Thus the above theorem shows that the basic eigenvalue bound refutes a random $d$-regular instance of NAE-3SAT (whp) provided $d > 13.5$.

\section{SDP solutions for random instances}
As a guide for our construction, let us imagine SDP solutions for the Max-Cut problem on the infinite graph $\Primal_{d,c}$.  (As these imaginings are only for intuition's sake, we will not be completely formal.)
To lower bound $\SDP(\Primal_{d,c})$, it is necessary and sufficient to construct jointly standard Gaussian random variables~$(\bX_v)_{v \in V(\Primal_{d,c})}$ for which the correlation~$\E[\bX_{u}\bX_{v}]$ --- ``on average'', over all edges $\{u,v\} \in E(\Primal_{d,c})$ --- is very negative.  It's simpler, and stronger, to look for such a Gaussian process in which $\E[\bX_u \bX_v] = \vr$ for \emph{every} edge~$\{u,v\}$, with $\vr$ as negative as possible.  Such solutions would give an upper bound for the Lov{\'a}sz theta value,  $\LTheta{\Primal_{d,c}} \leq 1-1/\vr$, while still giving an SDP lower bound of $\SDP(\Primal_{d,c}) \geq \frac12 - \frac12 \vr$.  In turn, we would have such a Gaussian process provided it satisfied
\begin{equation}    \label{eqn:infinite-eigenvalue}
    \phantom{\quad \text{for all } v \in V(\Primal_{d,c})}
    \frac{1}{\DEG}\sum_{u \sim v %\{u,v\} \in E(\Primal_{d,c})
    } \bX_u = \vr \bX_v \quad \text{for all } v \in V(\Primal_{d,c}),
\end{equation}
where, as before, $\DEG = (c-1)d$ is the degree of each~$v$.  This is the ``eigenvalue equation'' for $A_{\Primal_{d,c}}$ for $\lambda = \DEG\vr$.  Thus one may suspect that \Cref{eqn:infinite-eigenvalue} is possible whenever $\lambda = \DEG\vr \in \spec(A_{\Primal_{d,c}})$.  Given $\spec(A_{\Primal_{d,c}})$ as in \Cref{eqn:infinite-spectra}, we may therefore hope to obtain the desired Gaussian process for any
\begin{equation}    \label{eqn:cf}
    \vr \in \bracks*{\frac{\llower^2 - d}{\DEG}, \frac{\lupper^2 - d}{\DEG}} = \bracks*{1 - \frac{(1+\rho_1)^2}{\DEG}, 1 - \frac{(1-\rho_1)^2}{\DEG}};
\end{equation}
in particular, for the most negative such value,
\begin{equation}    \label{eqn:vrstar}
    \vrstar = 1 - \frac{(1+\rho_1)^2}{\DEG}.
\end{equation}
This would lead to the lower bound
\[
    \SDP(\Primal_{d,c}) \geq \frac12 - \frac12 \vrstar = \frac{(1+\rho_1)^2}{2\DEG}.
\]
In fact, since $\Primal_{d,c}$ is a vertex-transitive graph, it follows from a theorem of Harangi and Vir{\'a}g that such Gaussian processes do exist, and they can be constructed in a simple fashion as ``linear block factors of IIDs'':
\begin{theorem}  (\cite[Theorem~4]{HV15}.)                                   \label{thm:harangi-virag}
    Let $G$ be an infinite vertex-transitive graph with adjacency operator~$A_G$. Then for each $\lambda \in \Spec(A_G)$, there is an $\mathrm{Aut}(G)$-invariant standard Gaussian process $(\bX_v)_{v \in V(G)}$ for which $\sum_{u \sim v} \bX_u = \lambda \bX_v$ holds for all $v \in V(G)$.  Furthermore, the process can be approximated (in distribution) by a ``linear block factor of IID process'', meaning one that is constructed as follows:  $(\bZ_v)_{v \in V(G)}$ are chosen as IID standard Gaussians, and then $\bX_v$ is set to be a fixed linear function~$f$ of those $\bZ_u$'s which have $\dist_G(u, v) \leq L$, where $L$ is a finite ``radius''.
\end{theorem}
As mentioned in \Cref{sec:prior}, results of this nature date back at least to the work of Elon~\cite{Elo09}, who constructed such ``Gaussian waves'' on the infinite $d$-regular tree~$\Tree_d$.  An important aspect of \Cref{thm:harangi-virag} is the ``block'' aspect, meaning that each~$\bX_v$ is defined just from a ``local'', finite number of~$\bZ_u$'s.  Thus we can hope to use the construction for (primal graphs of) large but finite $(c,d)$-biregular graphs with large girth, which locally look tree-like.

That said, we cannot quite use the \Cref{thm:harangi-virag} as a black box for our purposes, for a few reasons.  One reason is that we want to apply it to large random biregular graphs, which will not strictly speaking have low girth, but will merely have ``few'', ``far apart'' short cycles.  Second, we will be constructing SDP solutions for \emph{edge-signed} graphs, a slight generalization of \Cref{thm:harangi-virag}'s framework.  Finally, it will be nice for us to reason about $\E[\bX_u \bX_v]$ not just for adjacent $u$,~$v$.

On the other hand, the construction of the linear block factor of IID process for $\Primal_{d,c}$ is a fairly straightforward generalization of earlier concrete constructions for~$\Tree_d$ such as the one in~\cite{CGHV15}.  We present it in the next section.

\subsection{Linear factors of IIDs}
Here we essentially prove  \Cref{thm:harangi-virag} in the special case of $\Primal_{d,c}$. The proof closely follows~\cite[Section~3]{CGHV15}.
\begin{theorem}  \label{thm:idealized-FIID}
    Let $c,d \geq 2$ and let $\lambda \in \spec(A_{\Primal_{d,c}})^\circ = (\llower^2 - d, \lupper^2 - d)$.  Then there exist~$L \in \N$ and reals $a_0, a_1, \dots, a_L$ such that the following holds:  When $(\bZ_v)_{v \in V(\Primal_{d,c})}$ are IID standard Gaussians, and the random variables $(\bX_v)_{v \in V(\Primal_{d,c})}$ are formed via
    \begin{equation}    \label{eqn:the-FIID}
        \bX_v = \sum_{\ell = 0}^L \sum_{\substack{w \in V(\Primal_{d,c}) \\ \dist(w,v) = \ell}} a_\ell \bZ_w,
    \end{equation}
    then we have $\E[\bX_v^2] = 1$ for all $v$ (so that the $\bX_v$'s are jointly standard Gaussians), and ${\E[\bX_u\bX_v] = \frac{\lambda}{\DEG}}$ for all $\{u,v\} \in E(\Primal_{d,c})$.  In other words (cf.~\Cref{eqn:cf}):
    \begin{equation}    \label{eqn:finish}
         \text{for any} \quad 1 - \frac{(1+\rho_1)^2}{\DEG} < \vr < 1 - \frac{(1-\rho_1)^2}{\DEG} \quad \text{we can achieve $\E[\bX_u\bX_v] = \vr$} \quad \forall \{u,v\} \in E(\Primal_{d,c}).
    \end{equation}
\end{theorem}
\begin{proof}
    Let us temporarily relax the requirement that $L$ be finite.  To that end, we will consider defining
    \begin{equation}    \label{eqn:X-FIID}
        \bX_v = \gamma \cdot \sum_{\ell=0}^\infty \ \sum_{\substack{w \in V(\Primal_{d,c}) \\ \vdist(w,v) = \ell}} r^\ell \bZ_w,
    \end{equation}
    for constants $\gamma \in \R^+$, $r \in \R$.
    It follows that for two vertices $u, v \in V(\Primal_{d,c})$ with $\dist(u,v) = h$, we have
    \begin{equation}            \label{eqn:general-corr}
        \E[\bX_u \bX_v] = \gamma^2 \cdot \sum_{j,k = 0}^\infty p^h_{j,k} r^{j+k}.
    \end{equation}
    In this proof we focus only on $h = 0, 1$, saving $h > 1$ for \Cref{thm:full-FIID}. By \Cref{prop:pijk} we have
    \[
        \#\braces{w : \vdist(w,v) = \ell} = p^0_{\ell,\ell} =
                                                   \begin{cases}
                                                        1 & \text{if $\ell = 0$,} \\
                                                        (\rho_1^2 + \csq^2) \cdot \rho_1^{2(\ell-1)} & \text{if $\ell > 0$,}
                                                   \end{cases}
    \]
    where recall $\rho_1^2 + \csq^2 = (c-1)d$ and  $\rho_1^2 = (c-1)(d-1)$. Thus
    \begin{equation}    \label{eqn:the-variance}
        \E[\bX_v^2] = \Var[\bX_v] = \gamma^2 \cdot \parens*{1 + \sum_{\ell=1}^\infty (\rho_1^2 + \csq^2) \cdot \rho_1^{2(\ell-1)} \cdot r^{2\ell}} = \gamma^2 \cdot \frac{1+(\csq r)^2}{1-(\rho_1 r)^2}, \quad \text{provided $|r| < \rho_1^{-1}$.}
    \end{equation}
    By choosing $\gamma$ such that
    \[
        \gamma^2 = \frac{1-(\rho_1 r)^2}{1+(\csq r)^2}
    \]
    we get $\Var[\bX_v] = 1$.  On the other hand, for fixed $u, v$ with $\vdist(u,v) = 1$ we have
    \[
        \#\braces*{w : \vdist(u,w) = \ell_1, \vdist(v,w) = \ell_2} = p^1_{\ell_1,\ell_2} =
                \begin{cases}
                        (\csq^2-1) \cdot \rho_1^{2(\ell - 1)} & \text{if $\ell_1 = \ell_2 > 0$,} \\
                        \rho_1^{2\ell_1} & \text{if $\ell_2 = \ell_1 + 1$,} \\
                        \rho_1^{2\ell_2} & \text{if $\ell_1 = \ell_2 + 1$,} \\
                        0 & \text{else,}
                \end{cases}
    \]
    where recall $\csq^2 -1 = c-2$. Thus
    \begin{equation}    \label{eqn:the-corr}
        \E[\bX_u \bX_v] = \gamma^2 \cdot \parens*{\sum_{\ell = 1}^\infty (\csq^2-1)\cdot \rho_1^{2(\ell-1)}\cdot r^{2\ell} + \sum_{\ell = 0}^\infty 2\cdot \rho_1^{2\ell} \cdot r^{2\ell+1}} = \gamma^2 \cdot \frac{1+(\csq r)^2 - (1-r)^2}{1-(\rho_1 r)^2},
    \end{equation}
    and so by our choice of $\gamma$ we conclude
    \[
        \E[\bX_u \bX_v]  = 1 - \frac{(1-r)^2}{1+(\csq r)^2}.
    \]
    Calculus shows that the expression on the right is increasing for $r$ in the range $[-\csq^{-2}, 1]$, which is a superset of the range that \Cref{eqn:the-variance} allows us for~$r$, namely $(-\rho_1^{-1}, \rho_1^{-1})$.  This establishes \Cref{eqn:finish}; the only catch is that we haven't used a finite~$L$.  But this can be achieved by truncating the sum in \Cref{eqn:X-FIID} to $\ell \leq L$ for $L$ sufficiently large.  This truncation only changes \Cref{eqn:the-variance,eqn:the-corr} by a quantity that decays like~$(\rho_1 r)^L$.  Thus the change in $\E[\bX_u \bX_v]$ from truncation can be made arbitrarily small, and this is acceptable for the conclusion \Cref{eqn:finish} because the desired interval of~$\vr$'s is open.
\end{proof}
\begin{corollary}                                       \label{cor:signed-FIID}
    \Cref{thm:idealized-FIID} also holds for the primal graph~$\mathbb{I}$ of any edge-signed version $\mathbb{X}$ of $\Tree_{d,c}$ (as defined in \Cref{sec:graph-notation}), in the sense of having $\E[\bX_u\bX_v] = \xi_{uv} \vr$ for all $\{u,v\} \in E(\mathbb{I})$, where $\xi_{uv}$ denotes the sign of edge $\{u,v\}$.
\end{corollary}
\begin{proof}
    Assume we have signs $\xi_{av} \in \{\pm 1\}$ for each constaint/variable edge $\{a,v\}$ in $\mathbb{X}$, and therefore signs $\xi_{uv} = \xi_{au}\xi_{av}$ for each edge $\{u,v\}$ in $\mathbb{I}$.  It's clear that for any closed walk in the tree~$\mathbb{X}$, the product of the edge-signs along the walk is~$1$; by construction, it follows that the same is true in~$\mathbb{I}$.  Thus for any $u,v \in V(\mathbb{I})$ (not necessarily adjacent) we can unambiguously define $\xi[u \leftrightarrow v]$ as the product of edge-signs along any $uv$-path in~$\mathbb{I}$. We now alter the construction in \Cref{eqn:X-FIID} as follows:
    \[
        \bX_v = \gamma \cdot \sum_{\ell=0}^\infty \ \sum_{\substack{w \in V(\Primal_{d,c}) \\ \vdist(w,v) = \ell}} \xi[w \leftrightarrow v] r^\ell \bZ_w,
    \]
    Clearly $\Var[\bX_v]$ is unchanged.  As for $\E[\bX_u \bX_v]$, the contribution from each $\bZ_w$ now yields an additional factor of $\xi[w\leftrightarrow u]\xi[w \leftrightarrow v] = \xi[u \leftrightarrow v] = \xi_{uv}$.  Thus each $\E[\bX_u \bX_v]$ changes by a factor of $\xi_{uv}$, as desired.  The rest of the proof is the same.
\end{proof}
\begin{theorem}                                       \label{thm:full-FIID}
    In the $L = \infty$ setting of  \Cref{thm:idealized-FIID}, we in fact obtain, for all $r \in (-\rho_1^{-1}, \rho_1^{-1})$ and all $u,v \in V(\Primal_{d,c})$,
    \[
        \E[\bX_u \bX_v] =  r^h\parens*{1 + \frac{h(1-r)(1 + \csq^2 r)}{1+(\csq r)^2}}, \quad \text{where $h = \dist(u,v)$.}
    \]
    (The $r = 0$ case is of course trivial, with $\bX_v = \bZ_v$.)
\end{theorem}
\begin{proof}
    Allowing $L$ to be infinite and returning to \Cref{eqn:general-corr}: for $u,v \in V(\Primal_{d,c})$ with ${\dist(u,v) = h}$, one can use \Cref{prop:pijk} to show (calculations omitted) that
    \[
        \E[\bX_u\bX_v] = \gamma^2 \cdot \frac{r^h(1+(\csq r)^2 + h(1-r)(1 + \csq^2 r))}{1-(\rho_1 r)^2}
    \]
    provided $|r| < \rho_1^{-1}$. The result follows.
\end{proof}
\begin{remark}  \label{rem:for-triangle-ineqs}
    One can show that the expression in \Cref{thm:full-FIID} has the property that its absolute value is a strictly decreasing function of~$h$ for every $r \neq 0$. (Indeed, it decreases exponentially.)  This is the key takeaway of the theorem, implying that in the setting of \Cref{cor:signed-FIID}, $\abs{\E[\bX_u \bX_v]} \leq \abs{\vr}$ for \emph{all} distinct pairs $u,v \in \mathbb{I}$ (with equality when $\{u,v\} \in E(\Primal_{d,c})$).
\end{remark}

\subsection{SDP solutions for randomly lifted/signed graphs}
In this section, let us fix $d \geq c \geq 2$, a small $\eps > 0$,
\[
    \vr = 1 - \frac{(1+\rho_1)^2}{\DEG} + \eps,
\]
and an~$L = L(\eps, c, d)$ such that \Cref{thm:idealized-FIID} and \Cref{cor:signed-FIID} hold.  Since each $\bX_v$ constructed therein depends only on the $\bZ_v$'s at distance at most~$L$ in $\Primal_{d,c}$ (and hence distance at most $2L$ in $\Tree_{d,c}$), we see that the exact same construction works equally well on any finite primal graph constructed from a $(c,d)$-biregular graph of girth exceeding~$4L$.  Thus (using also \Cref{rem:for-triangle-ineqs}) we immediately obtain:
\begin{theorem}                                     \label{thm:sdp-for-girth1}
    Let $H$ be any edge-signed $(c,d)$-biregular graph of girth exceeding $4L$ and let $I$ be its associated primal graph, with edge signs $\xi_{uv}$, $\{u,v\} \in E(I)$.  Then one can assign joint standard Gaussians $\bX_v$ to the vertices $v \in V(I)$ such that $\E[\bX_u \bX_v] = \xi_{uv} \vr$ for each edge $\{u,v\} \in E(I)$.  Furthermore, $\abs{\E[\bX_u \bX_v]} \leq \abs{\vr}$ for all distinct $u,v \in V(I)$.  As consequences:
    \begin{enumerate}[label=(\roman*)]
        \item \label{item:lovasz} If $H$ is unsigned, $\LTheta{I} \leq 1-1/\vr$.
        \item \label{item:2xor} If we view $I$ as a 2XOR-SAT instance, we have $\SDPtri(I) \geq \half - \half \vr = \frac{(1+\rho_1)^2}{2\DEG} - \eps$.
        \item \label{item:3nae} If $c = 3$ and we view $I$ as a $d$-regular NAE-3SAT instance, we have
                $
                    \SDPtri(I) \geq \frac13 - \frac13 \vr = \frac98 - \frac38\cdot\frac{\parens*{\sqrt{d-1} - \sqrt{2}}^2}{d} - \eps.
                $
    \end{enumerate}
\end{theorem}

We have the following corollary:
\begin{theorem}                                     \label{thm:sdp-for-girth2}
    Let $Y$ be a $(c,d)$-biregular bipartite graph and let $\bY_n$ be a random $n$-lift of~$Y$.  Let~$\bH_n$ denote an \emph{arbitrary} edge-signing of~$\bY_n$, and $\bI_n$ its associated primal graph.  Then:
    \begin{enumerate}
        \item With positive probability (depending only on~$d$ and~$\eps$), \Crefrange{item:lovasz}{item:3nae} of \Cref{thm:sdp-for-girth1} all hold.
        \item With high probability, \Cref{item:2xor,item:3nae} of \Cref{thm:sdp-for-girth1} hold with an additive loss of~$O(1/n)$.
    \end{enumerate}
    \end{theorem}
    \begin{proof}
    The first statement is an immediate consequence of \Cref{thm:girthy}.  As for the second statement, \Cref{thm:cycleless-neighborhoods} and Markov's inequality imply that, with high probability, only an $O((d+1)^{2L+2})/n = O(1/n)$ fraction of vertices in~$\bY_n$ are ``$(2L+2)$-bad'' (i.e., have a cycle within their distance-$(2L+2)$ neighborhood).  Assuming this holds, we use the linear block factors of IID solution from \Cref{thm:idealized-FIID} and \Cref{cor:signed-FIID} but with a small twist:  For each vertex~$v$ that is $2L$-bad in $\bY_n$, rather than using \Cref{eqn:the-FIID} we simply set $\bX_v = \bZ'_v$, where the random variables $\bZ'_v$ are new standard Gaussians independent of all other random variables.  Now for the $1-O(1/n)$ fraction of ``$(2L+2$)-good'' vertices, all their neighbors are still $2L$-good and thus are using the linear block factors of IID solution.   We therefore still have $\E[\bX_u \bX_v] = \xi_{uv} \vr$ for each edge $\{u,v\} \in E(I)$ where $u$ or $v$ is $(2L+2)$-good.  Furthermore, we still have $\abs{\E[\bX_u \bX_v]} \leq \abs{\vr}$ for all distinct $u,v \in V(I)$, since $\E[\bX_u \bX_v] = 0$ when one of $u$ or $v$ is $2L$-bad.  The second statement in the theorem therefore follows.
\end{proof}

\section{Conclusions}

In this work we have shown a sharp threshold for the SDP-satisfiability of random $d$-regular NAE-3SAT instances in the model of random lifts.  Some open questions that remain are the following:
\begin{itemize}
    \item Can we show similar sharp threshold results in the configuration model?  The main challenge is proving Friedman-style bounds on the spectra of random $(c,d)$-biregular bipartite graphs in this model.  An advantage to doing this would be the potential to show similar sharp thresholds for  $2$-coloring random $d$-regular $3$-uniform hypergraphs (i.e., random $d$-regular NAE-3SAT \emph{without} negations).
    \item Can we show similar sharp threshold results in the Erd\H{o}s--R{\'e}nyi random model?
    \item Can our analysis of the 2XOR-SAT SDP / Lov\'{a}sz theta function for the infinite biregular tree $\Tree_{d,c}$, and its primal graph $\Primal_{d,c}$ be extended to other interesting classes of infinite graphs (say, vertex-transitive)?  Are there application to other finite CSPs?
    \item A difficult but important open question: can we analyze the performance higher-degree ``Sum of Squares'' relaxations for refuting random sparse CSPs (that do not support pairwise-uniform distributions)?  Even analyzing the degree-$4$ Sum of Squares relaxation for NAE-3SAT or graph $3$-colorability seems very challenging.
\end{itemize}

\section*{Acknowledgments}
This work began at the American Institute of Mathematics workshop ``Phase transitions in randomized computational problems''; the authors would like to thank AIM, as well as the organizers Amir Dembo, Jian Ding, and Nike Sun, for the invitation. R.~O.~would like to thank Charles Bordenave, Sidhanth Mohanty, Doron Puder, Nike Sun, and David Witmer for helpful comments.

\bibliographystyle{alpha}
\bibliography{csp}

\appendix

\section{Bordenave's Theorem for random signed lifts}   \label{sec:bordenave}

% !TEX root = main.tex

\newcommand{\ve}{\vec{e}}
\newcommand{\bve}{\vec{\be}}

In this appendix we will prove the following theorem.
\begin{theorem*}[Restatement of \Cref{thm:signed-bordenave}]
    Let $Y$ be a connected graph (with more edges than vertices) having non-backtracking matrix $B$.  Fix $\eps > 0$.  Let $\bX_n$ be a random signing of a random $n$-lift $\bY_n$ of $Y$, and let $\bB_n$ denote the  non-backtracking matrix of $\bX_n$.   Then
    \[
        \Pr[\specrad(\bB_n) \geq \sqrt{\specrad(B)} + \eps] = o_{n \to \infty}(1).
    \]
\end{theorem*}
Our theorem requires minor modifications to the trace-method proof of~\cite[Theorem 20]{Bor17}, and we follow it closely.
The differences occur because~\cite[Theorem 20]{Bor17} pertains to the spectrum of unsigned lifts, and for that reason the arguments therein must take into account the uninteresting top eigenspace of the non-backtracking matrix; this introduces some technical complications.
Since we are working with randomly signed edges, we need not worry about these eigenspaces, and our arguments will be somewhat pared down (though to our knowledge they cannot be extracted from~\cite{Bor17} in a black-box fashion).

\subsection{Setup and notation}

We set the stage for the proof by introducing some notation and definitions.
Let $Y = (V,E)$ be an undirected graph, and let $\vec{E}$ be the set of directed edges associated with $E$, so that
\[
    \vec{E} = \{(u,v) : \{u,v\} \in E\},
\]
and $|\vec{E}| = 2|E|$.
To limit confusion, we will use plain, bold letters $\be$ to denote edges in $E$ and decorated bold letters $\vec{\be}$ to denote arcs in $\vec{E}$.
For an arc $\vec{\be} = (u,v)$, we let $(\vec{\be})^{-1} = (v,u)$.

Let $n \in \Z^+$, let $Y_n = (V_n, E_n)$ be an $n$-lift of $Y$ as defined in \Cref{sec:lift-notation}, and let $X_n = (V_n,E_n, \xi_n)$ be random signing of $Y_n$ with signs $\xi_n: E_n \to \R$.\footnote{In our setting, we will choose $\xi_n(e) \in\{ \pm 1\}$ independently and uniformly for each $e \in E_n$.}
In the $n$-lift, each edge $e \in E_n$ (arc $\vec{e} \in \vec{E_n}$) is associated with an edge $\{u,v\} \in E$ (arc $(u,v) \in \vec{E}$), and with a pair of labels $i,j \in [n]$, so that $e = \{(u,i),(v,j)\}$ ($\vec{e} = ((u,i),(v,j)$).
Again to limit confusion, we will use non-bold, plain letters to denote edges $e \in E_n$ and decorated, non-bold letters to denote arcs $\vec{e}\in \vec{E_n}$.
We let $S^{E}_n$ be the set of tuples of $|E|$ permutations on $[n]$.
Each $n$-lift is associated with some $\sigma = \{\sigma_{\be}\}_{\be \in E} \in S^{E}_n$, so that $E_n = \{ \{(u,i),(v,\sigma_{u,v}(i))\}\}$ (where we take $u$ to proceed $v$ lexicographically, in order to ensure that the bijection between $\sigma$ and lifts is unique).\footnote{Again, in our setting we will choose each $\sigma_{\be}$ uniformly at random in $S_n$.}
We sometimes refer to the lift specified by $\sigma \in S_n^E$ as $Y_n(\sigma)$.

We also define $B_n$ to be the weighted non-backtracking matrix of $X_n$ as in \Cref{sec:matrix-notation}, so that for directed edges $(u,v),(x,y) \in \vec{E_n}$,
\[
B_n[(u,v),(x,y)] = \xi_n(\{u,v\}) \cdot \Ind[v = x] \cdot \Ind[u \neq y].
\]

We will apply the trace method to $B_n$; that is, we will relate $\specrad(B_n)$ to the expected trace of a power of $B_n$.
\begin{fact}\label{fact:trace-method}
    If $A \in \C^{n\times n}$ is a random complex matrix, $m,\ell \in \Z^+$, $\eps,c \in \R^+$, and $\E[\tr((A^\ell(A^\ell)^*)^{k})] \le R^{2m\ell}$, then for $\ell\cdot m \ge \frac{c}{\eps} R \log n$ and $\eps < R/2 $,
    \[
	\Pr[\rho(A) \ge R + \eps ] \le n^{-c}
    \]
\end{fact}
\begin{proof}
    This follows by noticing that $\specrad(A)^\ell \le \sup_{x \in \R^n}\frac{\|A^\ell x\|_2}{\|x\|_2} = \| A^{\ell} (A^{\ell})^* \|^{1/2}$, and then applying Markov's inequality:
    \begin{align*}
	    \Pr[ \| A^{\ell} (A^{\ell})^* \|^{1/2\ell} \ge t]
	    \le \frac{\E[\tr((A^{\ell} (A^{\ell})^*)^{m})]}{t^{2m\ell}}
	    \le \left(\frac{R}{t}\right)^{2m\ell},
	\end{align*}
	and choosing $t = R + \eps$ with $2\eps < R$,
	\[
	    \left(\frac{1}{1 + \eps/R}\right)^{2m\ell} \le \left(1 - \frac{\eps}{2R}\right)^{2m\ell} \le \exp\left(- \frac{\eps m \ell}{R}\right)
	    \]
	    for $\ell \cdot m \ge  \frac{c}{\eps} R\log n$ the conclusion follows.
\end{proof}

In our computations, we will bound the contribution of sequences of \emph{half-edges} (so as to be consistent with \cite{Bor17}).
\begin{definition}[half-edge]
    A {\em half-edge} $\gamma$ is given by an arc $(u,v) \in \vec{E}$, and an index $i \in [n]$ corresponding to the index of $u$.
We think of $\gamma = ((u,v),i)$ as an arc leaving the $i$th copy of $u$ in the lift, and going to vertex $v$ at some unspecified index; colloquially, $\gamma = ((u,i), (v,?))$.

    We call the set of all possible half-edges $\Pi$.
    In the interest of promoting clarity, we point out that $\Pi$ does not depend on the specific choice of lift, $\sigma$.
\end{definition}

\begin{definition}[valid sequence of half-edges]
We will say that a sequence of half-edges $(\gamma_1,\ldots,\gamma_{2k})$ is {\em valid} if it satisfies the following constraints:
\begin{enumerate}
    \item {\em Admissibility of pairs}: consecutive pairs of half-edges correspond to the same edge in $Y$.
	Formally, for each $t \in [k]$ with $\gamma_{2t -1} = (\vec{\be}_{2t-1}, i_{2t-1})$ and $\gamma_{2t} = (\vec{\be}_{2t},i_{2t})$, we have that $\vec{\be}_{2t-1} = (\vec{\be}_{2t})^{-1}$.
    \item {\em Consistency}: if two half-edges are paired once, they remain paired for the remainder of the sequence.
	Formally, if there exists $t^*$ such that the half-edge $g = \gamma_{2t^*-1}$ is succeeded by the half-edge $h = \gamma_{2t^*}$, then for all $t$ such that $\gamma_{2t-1} = g$, we must also have $\gamma_{2t} = h$.
    Similarly, for all $t$ with $\gamma_{2t} = g$, we must also have $\gamma_{2t-1} = h$.
    \item {\em Consecutiveness}: the sequence of half-edges, when glued together, must correspond to a valid walk.
	Formally, for every $t$, if we have $\gamma_{2t} = ((u_{2t},v_{2t}), i_{2t})$ and $\gamma_{2t+1} = ((v_{2t+1}, u_{2t+1}),i_{2t+1})$, then we must have $v_{2t+1} = v_{2t}$ and $i_{2t+1} = i_{2t}$.
\end{enumerate}
\end{definition}

Colloquially, if two half-edges $\gamma= (\be,i),\gamma'=((\be)^{-1},j)$ appear consecutively in a sequence with $\gamma$ in an odd position and $\gamma'$ in an even position, we will say that they are {\em glued together} to give the edge $\{(\be_1,i),(\be_2,j)\}$ (where $\be_1,\be_2$ are the first and second endpoints of $\be$, respectively).

\begin{definition}[non-backtracking sequence]
    A sequence of half-edges $(\gamma_1,\ldots,\gamma_{k})$ is called {\em non-backtracking} if it does not define a walk that backtracks; that is, for each $t \in [k]$, if $\gamma_{2t} = (\be_{2t}, i_{2t})$ and $\gamma_{2t+1} = (\be_{2t+1}, i_{2t+1})$, we require that $\be_{2t} \neq \be_{2t+1}$.
\end{definition}

We define $\Gamma^{2k}$ to be the set of all valid, non-backtracking sequences of $2k$ half-edges.

\subsection{Walk decomposition}

For $\be = \{u,v\} \in E$, define $M_{\be}$ to be the $n\times n$ signed permutation matrix which encodes $\sigma_e$, so that $(M_{\be})_{ij} = \xi(\{(u,i),(v,j)\})$ if and only if $\sigma_{\be}(i) = j$.
Further, for two half edges $\gamma = (\vec{\be},i),\gamma' = (\vec{{\bf f}},j)$, we let $M_{\gamma,\gamma'} = \Ind[\vec{\be} = (\vec{{\bf f}})^{-1}] \cdot \Ind[\sigma_{\be}(i) = j] \cdot \xi((\be_1,i),(\be_2,j))$ (where $\be$ is the undirected version of $\vec{\be}$).

For two arcs $\vec{e},\vec{f} \in \vec{E}_n$, let $\Gamma_{\vec{e},\vec{f}}^{2k}$ be the set of all valid, non-backtracking sequences of $2k$ half-edges $(\gamma_1,\ldots,\gamma_{2k})$, such that $\gamma_1,\gamma_2$ form $e$ when glued together, with the direction of $\vec{e}$ specified by $\gamma_1$, and such that $\gamma_{2k-1},\gamma_{2k}$ form $\vec{f}$ when glued together, with the direction of $\vec{f}$ specified by $\gamma_{2k-1}$.
We have by definition that
\begin{align}
    (B_n^{k})_{ef} &= \sum_{\gamma \in \Gamma_{\vec{e},{\vec{f}}}^{2k+2}} \prod_{s=1}^k M_{\gamma_{2s-1}\gamma_{2s}}, \label{eq:sum}
\end{align}
since if a sequence $\gamma$ is not valid or non-backtracking, it will have value $0$.

We now define {\em tangles}, which are undesirable, low-probability walk structures (we will be able to discard their contribution to \Cref{eq:sum}).
\begin{definition}[tangle-free]
    For a positive integer $\ell$, a graph $G$ is { \em $\ell$-tangle free} if it contains at most one cycle in every neighborhood of radius at most $\ell$.
    A valid sequence $\gamma \in \Gamma^{2k}$ is $\ell$-tangle free if the graph given by the edges and vertices visited by $\gamma$ does not contain more than one cycle in any neighborhood of radius at most $\ell$.
\end{definition}

The following lemma from \cite{Bor17} proves that with high probability, $Y_n$ is $\ell$-tangle free.

\begin{lemma}[{\cite[Lemma 24]{Bor17}}]\label{lem:no-tangle}
    If $\ell \le \kappa \log_{d-1} n$ with $\kappa \in [0,1/4]$ and $d$ the maximum degree of a vertex in $Y$, then with high probability $Y_n$ is tangle-free.
\end{lemma}

Finally, we will require the following definition.
\begin{definition}
    A valid sequence $\gamma$ is {\em even} if the walk it induces contains every undirected edge with even multiplicity.
\end{definition}

\subsection{Bounding the expectation of a single walk}\label{sec:onewalk}

Now, we bound the expectation of the product of entries along a walk.

For a sequence $\gamma=  (\gamma_1,\ldots,\gamma_{2\ell})$ of length $2\ell$, with $\gamma_t = ((u_t,v_t),i_t)$, let $E_\gamma$ be the set of lifted edges in $\gamma$,
\[
    E_\gamma = \{ \{(u_{2t-1}, i_{2t-1}),(v_{2t-1},i_{2t})\} \mid t \in [k] \}.
\]
\begin{proposition}\label{prop:onewalk}
    Suppose that $\gamma$ is a valid sequence of length $2k \ll \sqrt{n}$.
    Let
    $\ell < \frac{1}{4} \log_{d-1} n$.
    Then we have
    \[
	\E_{\sigma,\xi} \left[\prod_{s=1}^{k} M_{\gamma_{2s-1}\gamma_{2s}} \right]
	\le \Ind[\gamma \text{ even}]\cdot(1+o_n(1))\cdot \left(\frac{1}{n}\right)^{|E_{\gamma}|}.
    \]
\end{proposition}
\begin{proof}
    Consider some valid sequence of half-edges $\gamma = (\gamma_1,\ldots,\gamma_{2k})$, and let $\gamma_t = ((u_t,v_t), i_t)$ and $\vec{\be_t} = (u_t,v_t)$, $\be_t = \{u_t,v_t\}$ for convenience.
    We have that
    \begin{align}
	\E_{\sigma,\xi} \left[\prod_{s=1}^{k} M_{\gamma_{2s-1}\gamma_{2s}}  \right]
	&= \E_{\sigma,\xi} \left[\prod_{\be \in \gamma} \prod_{\substack{t \in [k]\\ \be_{2t-1} =\be}} M_{\gamma_{2t -1}\gamma_{2t}}\right]
	= \prod_{\be \in \gamma} \E_{\sigma,\xi} \left[\prod_{\substack{t \in [k]\\ \be_{2t-1} =\be}} M_{\gamma_{2t -1}\gamma_{2t}}\right],\label{eq:edge-prod}
    \end{align}
    since for $\be \neq \be'$, $\sigma_{\be}$ and $\sigma_{\be'}$ are independent, and by the independence of $\xi_n$.
    Expanding the entries of $M$ according to $M$'s definition,
    \begin{align}
	\cref{eq:edge-prod}
	&= \prod_{\be \in \gamma} \E_{\sigma,\xi} \left[\prod_{\substack{t \in [k]\\ \be_{2t-1} = \be}} \Ind[\sigma_{e_{2t-1}}(i_{2t-1}) = i_{2t}] \cdot \xi((u_{2t-1},i_{2t-1}),(v_{2t-1},i_{2t})) \right].\label{eq:expand}
    \end{align}
    By the independence of the signing $\xi$, we have that the expectation of any sequence in which any (undirected) edge is visited an odd number of times is $0$.
    Assimilating this fact,
    \begin{align}
	\cref{eq:expand}
	= \Ind[\gamma \text{ even}] \cdot \prod_{e \in \gamma} \E_{\sigma} \left[\prod_{\substack{t \in [k]\\ \be_{2t-1} = \be}} \Ind[\sigma_{\be_{2t-1}}(i_{2t-1}) = i_{2t}] \right].\label{eq:unsign}
    \end{align}
    Now, suppose that $k_{\be}$ distinct lifted copies of the edge $\be \in E$ appear in $\gamma$.
    Since $\gamma$ is consistent, and because we may assume every edge appears with even multiplicity, the term within the expectation just corresponds to fixing $k_{\be}$ edges of a permutation on $n$ elements.
    Thus we simplify,
    \begin{align}
	\cref{eq:unsign}
	&= \Ind[\gamma \text{ even}] \cdot \prod_{\be \in \gamma} \frac{(n-k_{\be})!}{n!}
	\le \Ind[\gamma\text{ even}] \cdot \prod_{\be \in \gamma} \left(\frac{1}{n}\left(1 + \frac{2k_{\be}}{n}\right)\right)^{k_{\be}},
    \end{align}
    where to obtain the last inequality we have used that for $i \le k_{\be}\ll \sqrt{n}$,
    \[
	\frac{1}{n - i} \le \frac{1}{n}\left(1+\frac{2i}{n}\right) \le \frac{1}{n} \left(1 + \frac{2k_{\be}}{n}\right).
    \]
    And now since $\sum_{\be \in \gamma} k_{\be} = |E_\gamma|$ is the number of distinct lifted edges in $\gamma$, and the number of base edges is at most the number of lifted edges,
    \begin{align}
	\le \Ind[\gamma\text{ even}] \cdot \left(\frac{1}{n}\right)^{|E_\gamma|} \left(1 + \frac{2k}{n}\right)^{2k}.
    \end{align}
    Using that $2k \ll \sqrt{n}$ we obtain our conclusion.
\end{proof}

\subsection{Counting walks}
To apply \Cref{fact:trace-method}, we will need to bound the trace of a power of $B_n^{\ell}(B_n^{\ell})^*$.
Since the trace corresponds to a sum over walks, and because in \Cref{sec:onewalk} we have a bound on the expectation of each walk as a function of the number of distinct edges and the evenness of the walk, we have reduced our problem to counting the number of walks of various types.
We will follow the definitions of Bordenave rather closely, so we may recycle his bounds.

We have that
\begin{align}\label{eq:trace-full}
    \tr\left((B^{\ell}(B^\ell)^{*})^m\right)
    &= \sum_{\substack{e_1,\ldots,e_{2m-1} \in E_n^{2m-1}}} \prod_{s = 1}^{2m-1} (B_n^\ell)_{e_{s}, e_{s+1}},
\end{align}
where we have taken $s+1$ modulo $2m-1$.
To characterize the summation, it is useful for us to define the following set of sequences of half-edges, which have the property that large sub-sequences are tangle-free.

\begin{definition}
Let $W_{\ell,m}$ be the set of sequences of half-edges $\gamma$ of length $2\ell\times 2m$ with the properties that, if we write $\gamma$ as a sequence of sub-sequences $\gamma = (\gamma^{(1)},\ldots,\gamma^{(2m)})$
\begin{enumerate}
    \item For each $s \in [2m]$, the sub-sequence $\gamma^{(s)}$ is valid, non-backtracking, and {\bf tangle-free}.
    \item For each $s \in [m]$, the final edge in $\gamma^{(s)}$ is equal to the first edge in $\gamma^{(s+1)}$ (where we take addition mod $2m$).
	Formally, if $\gamma^{(t)} = (((u^{(t)}_1,v^{(t)}_1),i^{(t)}_1),\ldots,((u^{(t)}_{2\ell},v^{(t)}_{2\ell}),i^{(t)}_{2\ell}))$, then we require $u^{(s)}_{2\ell-1} = u^{(s+1)}_1$, $v^{(s)}_{2\ell-1} = v^{(s+1)}_1$, $i^{(s)}_{2\ell-1} = i^{(s+1)}_1$ and $i^{(s)}_{2\ell} = i^{(s+1)}_2$.
\end{enumerate}
\end{definition}

Recall we have defined $\Pi$ to be the set of all half-edges (not necessarily present in $Y_n$).
\begin{definition}
We define an equivalence relation on $\Pi^m$: $\gamma,\gamma' \in \Pi^m$, with $\gamma_t = ((u_t,v_t),i_t)$ and with $\gamma'_t = ((u'_t,v'_t),i'_t)$ for $t\in [m]$.
We'll say that for $\gamma \sim \gamma'$ if for all $t \in [m]$ we have $(u_t,v_t) = (u'_t,v'_t)$, and if in addition there exists a tuple of permutations in $S_n$, one for each vertex $u \in V$ from the base graph, $(\sigma_u)_{u \in V}$, so that $i'_t = \sigma_{u_t}(i_t)$.
\end{definition}

We observe that if $\gamma$ is {\em even}, then any $\gamma' \sim \gamma$ is even as well.
Similarly, if $\gamma \sim \gamma'$, then $|E_\gamma| = |E_{\gamma'}|$.
We choose a canonical representative for each equivalence class:

\begin{definition}[Canonical sequence]
    Let $V_\gamma(u) \subseteq \{u\} \times [n]$ be the set of all vertices of $Y_n$ visited by $\gamma$ which include $u$.
    We'll call $\gamma \in \Pi^m$ {\em canonical} if for all $u \in V$, $V_\gamma(u) = \{(u,1),\ldots,(u,|V_{\gamma}(u)|)\}$, and if the vertices of $V_{\gamma}(u)$ appear in lexicographical order in $\gamma$.
\end{definition}

The following lemmas are given in \cite{Bor17}.
\begin{lemma}[{\cite[Lemma 27]{Bor17}}]\label{lem:bord-bound-equiv}
    Let $\gamma \in \Pi^m$, and let $V_{\gamma} \subseteq V\times [n]$ be the set of vertices of $Y_n$ which appear in $\gamma$.
    Suppose that $|V_\gamma| = s$.
    Then $\gamma$ is isomorphic to at most $n^s$ elements in $\Pi^m$.
\end{lemma}

\begin{lemma}[{\cite[Lemma 28]{Bor17}}]\label{lem:bord-bound-cW}
    Let $\calW_{\ell,m}(s,a)$ be the subset of canonical paths in $W_{\ell,m}$ with $|V_\gamma| = s$ and $|E_\gamma| = a$.
    There exists a constant $\kappa$ depending on $\rho$ and $Y$ such that we have
    \[
	|\calW_{\ell,m}(s,a)| \le \rho^s (\kappa \ell m)^{8m (a-s+1) + 10m}.
	\]
\end{lemma}

We are now ready to bound the contribution of the sums of tangle-free sections.
\begin{proposition}\label{prop:full-bd}
    For $m = \lfloor \frac{\log n}{17 \log \log n}\rfloor$, $n \ge 3$, and $\ell \le \frac{1}{4}\log_{d-1} n$, and $1 < \rho$, there is a constant $c$ independent of $n$ such that
    \[
	\E\left[\sum_{\gamma \in W_{\ell,m}} \prod_{i=1}^{2m} \prod_{t=1}^{\ell} M_{\gamma^{(i)}_{2t-1},\gamma^{(i)}_{2t}}\right] \le n (c\ell m)^{10m} \rho^{(\ell + 2)m}.
    \]
\end{proposition}
\begin{proof}
    We split the left-hand side according to the $\calW$ equivalence classes,
    \begin{align}
	\E\left[\sum_{\gamma \in W_{\ell,m}} \prod_{i=1}^{2m} \prod_{t=1}^{\ell} M_{\gamma^{(i)}_{2t-1},\gamma^{(i)}_{2t}}\right]
	&\le \sum_{s=1}^{\infty} \sum_{a = s-1}^{\infty} n^s \sum_{\gamma \in \calW_{\ell,m}(s,a)} \E\left[\prod_{i=1}^{2m} \prod_{t=1}^{\ell} M_{\gamma^{(i)}_{2t-1},\gamma^{(i)}_{2t}}\right],\label{eq:bd}
    \end{align}
    where we have used that $|V_\gamma|-1 \le |E_\gamma|$, since $G_{\gamma}$ is connected.
    Now applying \Cref{prop:onewalk} (using that $\ell m \ll \sqrt{n}$), we have that for $\gamma \in \calW_{\ell,m}(s,a)$,
    \[
	\E\left[\prod_{i=1}^{2m} \prod_{t=1}^{\ell} M_{\gamma^{(i)}_{2t-1},\gamma^{(i)}_{2t}}\right] \le \Ind[\gamma \text{ even}]\cdot (1+o_n(1))\cdot \left(\frac{1}{ n}\right)^{a}.
    \]
Plugging this in above, along with the bound on $|\calW_{\ell,m}(s,a)|$ from \Cref{lem:bord-bound-cW}, we have
    \begin{align*}
	\text{\cref{eq:bd}}
	&\le \sum_{s=1}^{(\ell + 2)m + 1} n^s \sum_{a = s-1}^{(\ell+2)m} \rho^s (\kappa \ell m)^{8m(a - s + 1) + 10m} \cdot (1+o_n(1))\cdot\left(\frac{1}{n}\right)^a,
    \end{align*}
    where we use the fact that $\gamma$ must be even to obtain that $|E_\gamma| = s \le (\ell + 2)m$, (as there are only $2(\ell + 2)m$ edges in the sequence $\gamma$, and each must appear twice), and adjusted the upper limits of the summation accordingly.

    We re-index the above summation, setting $a' = a - s + 1$ and beginning to sum from $a' = 0$ (and summing till $a' = \infty$, as this yields a valid upper bound),
    \begin{align}
	\text{\cref{eq:bd}}
	&\le (1+o_n(1))\cdot(\kappa \ell m)^{10m} \cdot \sum_{s=1}^{(\ell + 2)m + 1} n^s\rho^s \cdot \left(\frac{1}{n}\right)^{s-1}\sum_{a' = 0}^{\infty} (\kappa \ell m)^{8ma'} \cdot \left(\frac{1}{n}\right)^{a'}\nonumber\\
	&= (1+o_n(1))\cdot n(\kappa \ell m)^{10m} \cdot \sum_{s=1}^{(\ell + 2)m + 1} \rho^s \sum_{a' = 0}^{\infty} (\kappa \ell m)^{8ma'} \cdot \left(\frac{1}{n}\right)^{a'}.\label{eq:blah}
    \end{align}
    For our chosen $m$, when $n$ is large enough, $\frac{(\kappa \ell m)^{8m}}{n} \le \frac{(\log n)^{16 m}}{n} \le n^{-1/17}$.
    Combining this observation with the fact that the rightmost sum is a geometric sum, there is a constant c such that
    \begin{align*}
	\cref{eq:blah}
	&\le cn(\kappa \ell m)^{10m} \cdot \sum_{s=1}^{(\ell + 2)m + 1} \rho^s.
    \end{align*}
    Finally, we are left again with a geometric sum; since we have $\rho > 1$, there is a constant $c'$ so that
    \begin{align*}
	&\le c'n(\kappa \ell m)^{10m} \cdot \rho^{(\ell + 2)m + 1}.
    \end{align*}
    Using that $\rho$ is independent of $n$ to push $\rho$ into the constant, we have our conclusion.
\end{proof}

\subsection{Putting things together}

We now finally have the ingredients to prove \Cref{thm:signed-bordenave}.

\begin{proof}[Proof of \Cref{thm:signed-bordenave}]
    Define $\rho := \specrad(B)$, fix $\eps > 0$, $\ell = \kappa \log_{d-1} n$ for a constant $\kappa \in (0,1/4)$, $m = \lfloor \frac{\log n}{17\log\log n}\rfloor$.
    By \Cref{lem:no-tangle}, if $\calE$ is the event that $Y_n$ is $\ell$-tangle-free,
    \[
	\Pr( \specrad(B_n) \ge \sqrt{\rho} + \eps)
	\le \Pr(\specrad(B_n) \ge \sqrt{\rho} + \eps, \calE) + o(1)
	\le \Pr(\|B_n^{\ell}(B_n^\ell)^*\|^{1/2\ell} \ge \sqrt{\rho} + \eps,\calE) +o(1).
    \]
    If $Y_n$ is $\ell$-tangle-free, then only sequences $\gamma \in W_{\ell,m}$ contribute to \Cref{eq:trace-full}, as any (consecutive) sub-sequence $\gamma^{(i)} \subset \gamma$ of length $2\ell$ defines a length-$\ell$ walk in $Y_n$.
So using \Cref{fact:trace-method} in conjunction with \Cref{eq:trace-full} and \Cref{prop:full-bd}, we have that
    \[
	\E[\tr((B_n^{\ell}(B_n^{\ell})^*)^m) \cdot \Ind[ \calE]]
	\le n(c\ell m)^{10m} \rho^{(\ell+2)m}.
    \]
Taking the $2\ell m$th root on the right, by our choice of $\ell = \Theta(\log n)$ and $m = \Theta(\log n/\log\log n)$, $(c\ell m)^{5/\ell} = o(\log^2 n)^{1/\log n} = 1 + o(1)$, $n^{1/2\ell m} \le 2^{\Theta(\log\log n/\log n)} = 1 + o(1)$, and since $\rho$ is independent of $n$, $\rho^{1/\ell} = 1 + o_n(1)$, and we have the desired conclusion.
\end{proof}

\end{document}